\documentclass[11pt,a4paper,english]{amsart}
\usepackage[utf8]{inputenc}
\usepackage[english]{babel} 

\usepackage{amsmath} %Numberwithin
\usepackage{amsthm} %Cambiar lo de teorema, proposición...
\usepackage{graphicx} 
\usepackage{xcolor} %Imprimir palabras de otro color
\usepackage{amsfonts}
\usepackage[biblabel]{cite}
\usepackage{bm} %Letra en negrita en math mode
\usepackage[hidelinks]{hyperref} %Poner links
\usepackage{amssymb} %Usar \supsetneq
\usepackage{tikz-cd} %Diagramas
\usepackage{amscd}
\usepackage{etoolbox}% http://ctan.org/pkg/etoolbox
\usepackage{breqn} %break lines with dmath
\patchcmd{\thmhead}{(#3)}{#3}{}{}
\usepackage[position=top]{subfig} %tables side by side
\usepackage{enumitem}%quitar indentación enumerate

\DeclareMathOperator{\ini}{in} %Inicial
\DeclareMathOperator{\ev}{ev} %Aplicación de evaluación
\DeclareMathOperator{\Tr}{Tr}

\DeclareMathOperator{\PRS}{PRS}
\DeclareMathOperator{\PRM}{PRM}
\DeclareMathOperator{\RM}{RM}
\DeclareMathOperator{\RS}{RS}

\DeclareMathOperator{\wt}{wt}

\newcommand{\F}{{\mathbb{F}}}
\newcommand{\fq}{\mathbb{F}_q}
\newcommand{\fqs}{\mathbb{F}_{q^s}}
\newcommand{\PP}{{\mathbb{P}}}
\newcommand{\PM}{{\mathbb{P}^{m}}}

\newcommand{\Z}{{\mathbb{Z}}}

\newcommand{\II}{{\mathfrak{I}}}

\newcommand{\T}{{\mathcal{T}}}
\newcommand{\zms}{{\mathbb{Z}^{m}_{q^s}}}
\newcommand{\zmsdos}{{\mathbb{Z}^{2}_{q^s}}}

\newcommand{\funf}{f_{a_2}^r}
\newcommand{\fung}{g_{a_2}^r}
\newcommand{\fungg}{g'_{a_2}}
\newcommand{\funh}{h_{a_2}^r}
\newcommand{\funl}{l_{a_2}}
\newcommand{\funll}{l'_{a_2}}
\newcommand{\condicion}{\bigcup_{c_2\in \II_{a_2},c_2>d-(q^s-1)}\II_{(d-c_2,c_2)}\subset \Delta_{\leq d}}
\newcommand{\condicionn}{\bigcup_{c_2\in \II_{a_2}}\II_{(d-c_2,c_2)}\subset \Delta_{\leq d}}
\newcommand{\condicionnn}{\bigcup_{c_2\in \II_{b_2},c_2>d-(q^s-1)}\II_{(d-c_2,c_2)}\subset \Delta_{\leq d}}
\newcommand{\notcondicion}{\bigcup_{c_2\in \II_{a_2},c_2>d-(q^s-1)}\II_{(d-c_2,c_2)}\not\subset \Delta_{\leq d}}

\usepackage{mathtools} %Valor abs y norma. Con * no se ajusta la altura.
\DeclarePairedDelimiter\abs{\lvert}{\rvert}%
\DeclarePairedDelimiter\norm{\lVert}{\rVert}%

\makeatletter
\let\oldabs\abs
\def\abs{\@ifstar{\oldabs}{\oldabs*}}
\let\oldnorm\norm
\def\norm{\@ifstar{\oldnorm}{\oldnorm*}}
\makeatother

\newtheorem{thm}{Theorem}[section]
\newtheorem{prop}[thm]{Proposition}
\newtheorem{cor}[thm]{Corollary}
\newtheorem*{cor*}{Corollary}
\newtheorem{lem}[thm]{Lemma}
\theoremstyle{definition}
\newtheorem{defn}[thm]{Definition} 

\newtheorem{rem}[thm]{Remark} 
 
\newtheorem{ex}[thm]{Example}

\title[Subfield subcodes of projective Reed-Muller codes]{Subfield subcodes of projective Reed-Muller codes}
\author{Philippe Gimenez, Diego Ruano and Rodrigo San-José}
\curraddr{%\texttt{Philippe Gimenez:}\\
\texttt{Philippe Gimenez, Diego Ruano, Rodrigo San-José:} IMUVA-Mathematics Research Institute, Universidad de Valladolid, 47011 Valladolid (Spain).
}
\email{pgimenez@uva.es;  diego.ruano@uva.es; rodrigo.san-jose@uva.es}
\date{}
\thanks{This work was supported in part by the following grants:
Grant TED2021-130358B-I00 funded by MCIN/AEI/10.13039/501100011033 and by the ``European Union NextGenerationEU/PRTR'', Grants PID2022-138906NB-C21 and PID2022-137283NB-C22 funded by MCIN/AEI/10.13039/501100011033 and by ERDF ``A way of making Europe'', and Grant FPU20/01311 funded by the Spanish Ministry of Universities.}



\begin{document}

\maketitle

\begin{abstract}
Explicit bases for the subfield subcodes of projective Reed-Muller codes over the projective plane and their duals are obtained. In particular, we provide a formula for the dimension of these codes. For the general case over the projective space, we generalize the necessary tools to deal with this case as well: we obtain a universal Gröbner basis for the vanishing ideal of the set of standard representatives of the projective space and we show how to reduce any monomial with respect to this Gröbner basis. With respect to the parameters of these codes, by considering subfield subcodes of projective Reed-Muller codes we obtain long linear codes with good parameters over a small finite field. 
\end{abstract}

\section{Introduction}

The subfield subcode of a linear code $C\subset \F_{q^s}^n$, with $s\geq 1$, is the linear code $C\cap \F_q^n$. This is a standard procedure that may be used to construct long linear codes over a small finite field. For instance, BCH codes can be seen as subfield subcodes of Reed-Solomon codes. In the multivariate case, the subfield subcodes of $J$-affine variety codes are well known \cite{galindo1} (in particular, the subfield subcodes of Reed-Muller codes) and have been used for several applications \cite{galindolcd,galindostabilizer}. The main problem that arises when working with subfield subcodes is the computation of a basis for the code, which also gives the dimension. In this paper, we compute bases for the subfield subcodes of projective Reed-Muller codes over the projective plane $\mathbb{P}^2$ and for their duals, and we also give tools to study the general case of projective Reed-Muller codes over the projective space $\mathbb{P}^m$. 

Projective Reed-Muller codes are evaluation codes obtained by evaluating multivariate homogeneous polynomials in the projective space. Arguing as in \cite{lachaud}, when one considers the sum of the rate and the relative minimum distance as a measure of how good the parameters of a code are, we obtain that projective Reed-Muller codes outperform Reed-Muller codes. It is therefore natural to pose the problem of studying the subfield subcodes of projective Reed-Muller codes, in particular, the problem of obtaining bases for the subfield subcode and its dual. As we stated previously, this has been done for different families of evaluation codes over the affine space \cite{galindo1,ssctoric}, but for evaluation codes over the projective space this has only been studied for evaluation codes over certain subsets of the projective line \cite{sanjoseSSCPRS}. In particular, the subfield subcodes of $J$-affine variety codes have been used for constructing quantum codes with good parameters \cite{galindo1,galindodistancestabilizer}, and one can expect that the subfield subcodes of projective Reed-Muller will also perform well in that setting.

In Section \ref{secp2}, we study the subfield subcode of a projective Reed-Muller code over the projective plane $\mathbb{P}^2$ and its dual. Comparing with projective Reed-Muller codes over $\PP^m$, with $m>2$, the case $m=2$ is usually the most interesting one because it can give rise to long codes with competitive parameters, which is similar to what happens in the affine case with Reed-Muller codes. For the case $m=2$, we provide explicit bases for both the subfield subcode of a projective Reed-Muller code over the projective plane $\mathbb{P}^2$ and its dual. In order to construct the basis for the dual, we consider Delsarte's Theorem \ref{delsarte}, which shows that we can generate the dual of the subfield subcode of a projective Reed-Muller code of degree $d$ by considering the evaluation of the traces of monomials of degree $d$. Then we can obtain a basis for the code by extracting a maximal linearly independent set of vectors, and we do this by using the vanishing ideal of the projective plane from Lemma \ref{vanishingidealP2} and the division by a Gröbner basis of this ideal. For the primary code, we study some polynomials obtained by combining traces of monomials and such that they can be regarded as homogeneous polynomials of degree $d$. We show that the set formed by their evaluations is linearly independent, and we conclude that this set is a basis for the code by a dimension argument, as we already have a basis of the dual code. 

We generalize some of the previous ideas to the general setting of the projective space $\mathbb{P}^m$ in Section \ref{secpm}. When we consider a larger $m$, we usually increase the length at the cost of having worse relative parameters, and also the analysis gets more complicated. Nevertheless, we are able to deal with this case as well. We give the vanishing ideal of a certain set of representatives of the points of $\mathbb{P}^m$. We prove that the set of generators that we give is a universal Gröbner basis of the ideal by using Buchberger's criterion \cite[\S 9 Thm. 3, Chapter 2]{cox} and showing that all the $S$-polynomials of the generators reduce to $0$, for any monomial order. From this result, we obtain the initial ideal and a basis for the quotient ring. Moreover, we provide a way to obtain the remainder of the division algorithm by this Gröbner basis for any monomial. This can be proved by checking that the remainder that we state is equivalent in the quotient ring to the original monomial, i.e., both have the same evaluation, and then checking that all the monomials in the support of the remainder are part of the basis given for the quotient ring. Particular cases of these ideas have been used previously for the projective line and the projective plane \cite{sanjoseSSCPRS,decodingRMP}, and we showcase them in full generality. With these tools, it is possible to deal with the general case of computing bases for the subfield subcodes of projective Reed-Muller codes over $\mathbb{P}^m$ and their duals, although getting explicit results as in the case $m=2$ seems out of reach as it gets too technical.

In Section \ref{secexamples}, we provide some examples of subfield subcodes of projective Reed-Muller codes. We compare their parameters with the codes from \cite{codetables}, and we see that some of the codes that we obtain have the best known parameters for the binary and ternary case. When considering longer codes, it is thus expected to also achieve good parameters, although the absence of tables for long codes makes comparisons difficult. One way to see that some of the longer codes also have good parameters is to consider the Gilbert-Varshamov bound \cite[Thm. 2.8.1]{huffman}. We provide a table with several of the codes that we obtain that exceed it. 

\section{Preliminaries}

We consider a finite field $\fq$ of $q$ elements with characteristic $p$, and its degree $s$ extension $\F_{q^s}$, with $s>1$. We consider the projective space $\PM$ over $\F_{q^s}$ and the polynomial ring $S=\fqs [x_0,\dots,x_m]$. Throughout this work, we will fix representatives for the points of $\PM$: for each point in $\PM$, we choose the representative whose first nonzero coordinate is equal to 1, starting from the left. We will denote by $P^m$ the set of representatives that we have chosen (seen as points in the affine space $\mathbb{A}^{m+1}$) and we will call them \textit{standard representatives}. Let $n=\abs{P^m}=\frac{q^{s(m+1)}-1}{q^s-1}$. We consider the following evaluation map:
$$
\ev_d:S_d \rightarrow \fqs^{n},\:\: f\mapsto \left(f(Q_1),\dots,f(Q_n)\right)_{Q_i \in P^m},
$$
where $S_d$ denotes the homogeneous polynomials of degree $d$. If $m=1$, the image of this evaluation map is the \textit{projective Reed-Solomon code} of degree $d$ (also called doubly extended Reed-Solomon code), and we will denote it by $\PRS_d$. The parameters of these codes are $[q^s+1,d+1,q^s-d+1]$. If $m>1$, then the image of the previous evaluation map is the \textit{projective Reed-Muller code} of degree $d$, which we will denote by $\PRM_d(m)$. This is another well known family of codes \cite{lachaud,sorensen}.

Given a code $C\subset \fqs^n$, its subfield subcode with respect to the extension $\fqs\supset \fq$ is defined as $C^\sigma:=C\cap \fq^n$. Subfield subcodes of projective Reed-Solomon codes, denoted by $\PRS_d^\sigma$, were studied in \cite{sanjoseSSCPRS}, and in this paper we are interested in studying the subfield subcodes of projective Reed-Muller codes and their dual codes, denoted by $\PRM_d^{\sigma}(m)$ and $\PRM_d^{\sigma,\perp}(m)$, respectively.  Before studying the projective case, let us show what happens in the affine case.

\subsection{Subfield subcodes of affine Reed-Muller codes}
The subfield subcodes of affine Reed-Muller, and, more generally, $J$-affine variety codes, are well known \cite{galindolcd,galindostabilizer}. We introduce now some of the basic techniques that are used to study the subfield subcodes of Reed-Muller codes, which we will denote by $\RM^\sigma_d(m)$.

Let $m\geq 1$ be an integer. We consider the ideal $I_{q^s}$ in the ring $R=\fqs[x_1,\dots,x_m]$ generated by $x_j^{q^s}-x_j$. It is clear that the finite set of points defined by $I_{q^s}$ is precisely the whole affine space $\mathbb{A}^m$ over $\fqs$. 

Let $n=q^{sm}$. Consider the quotient ring $R_{q^s}=R/I_{q^s}$ and the evaluation map $\ev_{\mathbb{A}^m}:R_{q^s}\rightarrow \F_{q^s}^{n}$ given by 
$$
\ev_{\mathbb{A}^m}(f)=(f(Q_1),f(Q_2),\dots,f(Q_{n}))_{Q_i\in \mathbb{A}^m}.
$$
This map is well defined and is clearly an isomorphism of vector spaces because $I_{q^s}$ is the vanishing ideal of $\mathbb{A}^m$. When working over quotient rings, we will use the same letter $f$ to denote the equivalence class and any polynomial representing it.

For $m=1$, the image by the evaluation map of $R_{\leq d}$, the polynomials of degree less than or equal to $d$, is the Reed-Solomon code of degree $d$ (sometimes called extended Reed-Solomon code), which we denote by $\RS_d$. For $m\geq 2$, the image by the evaluation map of $R_{\leq d}$ is the Reed-Muller code of degree $d$.

We introduce now multivariate cyclotomic sets, which are useful for computing the subfield subcodes of Reed-Muller codes. We consider $\mathbb{Z}/\langle q^s-1 \rangle$, where  we represent the classes of $\mathbb{Z}/\langle q^s-1 \rangle$ by $\{1,2,\dots,q^s-1\}$, and we define $\mathbb{Z}_{q^s}=\{0\}\cup \mathbb{Z}/\langle q^s-1 \rangle$, where we represent its classes by $\{0,1,\dots,q^s-1\}$. We will call a subset $\II$ of the Cartesian product $\mathbb{Z}^{m}_{q^s}:=\prod_{i=1}^m \mathbb{Z}_{q^s}$ a \textit{cyclotomic set} with respect to $q$ if $q\cdot c \in \II$ for any $c\in \II$.
Furthermore, $\II$ is said to be \textit{minimal} (with respect to $q$) if it can be expressed as $\II=\{q^i\cdot c,i=1,2,\dots\}$ for a fixed $c\in \II$, and in that situation we will write $\II_c:=\II$ and $n_c=\abs{\II_c}$. 

Now we define the following lexicographic order in the Cartesian product $\zms$: $a=(a_1,\dots,a_m) < (b_1,\dots,b_m)=b$ if and only if the rightmost entry of $b-a$, viewing this vector in $\mathbb{Z}^m$, is positive. We say that $a\in \II_c$ is a  \textit{minimal representative} of $\II_c$ if $a$ is the least element in $\II_c$ according to the order that we have given, and we will say that $b\in \II_c$ it is a \textit{maximal representative} of $\II_c$ if it is the biggest element. We will denote by $\mathcal{A}$ the set of minimal representatives of the minimal cyclotomic sets, and by $\mathcal{B}$ the set of maximal representatives of the minimal cyclotomic sets.

We can introduce a notion of degree for the elements in $\zms$. Given an integer $d\geq 1$, we define $\Delta_d=\{ c=(c_1,c_2,\dots,c_m)\in \zms \mid \sum_{i=1}^m c_i=d \}$, $\Delta_{<d}=\{ c=(c_1,c_2,\dots,c_m)\in \zms \mid \sum_{i=1}^m c_i<d \}$ and $\Delta_{\leq d}=\{ c=(c_1,c_2,\dots,c_m)\in \zms \mid \sum_{i=1}^m c_i\leq d \}$. We will also denote by $\mathcal{A}_{< d}$ and $\mathcal{A}_{\leq d}$ the elements $a \in \mathcal{A}$ such that $\II_a\subset \Delta_{<d}$ and $\II_a\subset \Delta_{\leq d}$, respectively.

\begin{ex}\label{exciclo}
Consider the extension $\F_4\supset \F_2$ with $m=2$. We have $q=2$ and $q^s=2^2=4$. Therefore, $\Z_{4}=\{0\}\cup \Z/\langle 3 \rangle$. We have the following minimal cyclotomic sets:
$$
\II_{(0,0)}=\{(0,0)\}, \II_{(1,0)}=\{(1,0),(2,0)\}, \II_{(0,1)}=\{(0,1),(0,2)\}, \II_{(1,1)}=\{(1,1),(2,2)\},
$$
$$
\II_{(3,0)}=\{(3,0)\}, \II_{(0,3)}=\{(0,3)\}, \II_{(3,3)}=\{(3,3)\}, \II_{(2,1)}=\{(2,1),(1,2)\},
$$
$$
\II_{(1,3)}=\{(1,3),(2,3)\}, \II_{(3,1)}=\{(3,1),(3,2)\}.
$$
The set of minimal representatives is $$\mathcal{A}=\{(0,0),(1,0),(0,1),(1,1),(3,0),(0,3),(3,3),(2,1),(1,3),(3,1) \},
$$
and the set of maximal representatives is:
$$
\mathcal{B}=\{(0,0),(2,0),(0,2),(2,2),(3,0),(0,3),(3,3),(1,2),(2,3),(3,2) \}.
$$
\end{ex}

For each $a\in \mathcal{A}$, we define the following trace map:
$$
\mathcal{T}_a:R_{q^s}\rightarrow R_{q^s},\:\: f\mapsto f+f^q+\cdots + f^{q^{(n_a-1)}},
$$
where we fix representatives in $R_{q^s}$ as follows: we will choose the representative of $f$ (and $\mathcal{T}_a(f)$) such that the monomials $x_1^{\gamma_1}\cdots x_m^{\gamma_m}$ in its support have their exponents reduced modulo $q^s-1$, i.e., $0\leq \gamma_i\leq q^s-1$, $1\leq i \leq m$. We will represent elements of $R_{q^s}$ and $R$ in the same way (simply as polynomials). Therefore, sometimes we consider $\mathcal{T}_a(f)$ as a polynomial in $R$ (the representative that we have chosen), which can be seen in other quotient spaces (such as the one we will define for the projective case).

\begin{ex}
Continuing with Example \ref{exciclo}, let us consider $a=(2,1)$ and compute $\T_a(x_1^2x_2)$. We have $n_a=2$ and, since $x_1^4=x_1$ in $R_4=\F_4[x_1,x_2]/\langle x_1^4-x_1,x_2^4-x_2\rangle$, then $\T_a(x_1^2x_2)=x_1^2x_2+x_1x_2^2$ which is the representative of  $x_1^2x_2+x_1^4x_2^2$ in $R_4$ with its exponents reduced modulo $q^s-1=3$.
\end{ex}

The following result gives a basis for the subfield subcodes of Reed-Muller codes (and also Reed-Solomon codes) \cite[Thm. 11]{galindolcd}, which we will denote by $\RM_d^\sigma(m)$.

\begin{thm}\label{baseafin}
Set $\xi_a$ a primitive element of the field $\F_{q^{n_a}}$. A basis for the vector space $\RM^\sigma_d(m)$ is obtained by considering the images under the map $\ev_{\mathbb{A}^m}$ of the set 
$$
\bigcup_{a\in \mathcal{A}_{\leq d}}\{\mathcal{T}_a(\xi_a^r x^a)\mid 0\leq r\leq n_a-1 \}.
$$
\end{thm}

As a consequence, we have that 
$$
\dim \RM_d^{\sigma}(m)=\sum_{a\in \mathcal{A}_{\leq d}}n_a.
$$

\begin{rem}\label{remafin}
Theorem \ref{baseafin} implies that, for different cyclotomic sets $\II_a\neq \II_b$, the evaluation of the polynomials in the sets $\{\mathcal{T}_a(\xi_a^r x^a)\mid 0\leq r\leq n_a-1 \}$ and $\{\mathcal{T}_b(\xi_b^r x^b)\mid 0\leq r\leq n_b-1 \}$ are linearly independent. Moreover, if we have $\II_a=\II_b$, then the previous sets generate the same vector space. 
\end{rem}

\subsection{Subfield subcodes of projective Reed-Muller codes}

Now we introduce the techniques that we will use to compute subfield subcodes of evaluation codes over the projective space. We had previously defined the usual evaluation map $\ev_d$ over the projective space, which can be generalized to the evaluation map $\ev:S\rightarrow \fqs^{n}$ given by
$$
\ev(f)=(f(Q_1),f(Q_2),\dots,f(Q_{n}))_{Q_i\in P^m}.
$$
It is clear that the kernel of the evaluation map is precisely the vanishing ideal of $P^m$, denoted by $I(P^m)$. If we consider $\ev(S_d)$ (corresponds to projective Reed-Solomon codes or projective Reed-Muller codes), the resulting code will be isomorphic to $S_d/(I(P^m)\cap S_d)\cong (S_d+I(P^m))/I(P^m)$. As we will see throughout this work, the vanishing ideal $I(P^m)$ gives plenty of information about these codes. 

\begin{rem}
Throughout the rest of the paper, given a set of polynomials $B$, we will refer to the set $\{\ev(f)\mid f \in B\}\subset \fqs^n$ as the evaluation of the set $B$.
\end{rem}

We will say that $f\in S$ evaluates to $\fq$ in $P^m$ if $\ev(f)\in \fq^n$. The following result gives us conditions for a polynomial to evaluate to $\fq$ in $P^m$.

\begin{lem}\label{lemassc}
One has that $f\in k[x_0,\dots,x_m]$ evaluates to $\fq$ in $P^m$ if and only if  $f(1,x_1,\dots,x_m)$, $f(0,1,x_2,\dots,x_m)$, $f(0,0,1,x_3,\dots,x_m)$,$\dots$, and $f(0,0,\dots,0,1,x_m)$ evaluate to $\fq$ in $\mathbb{A}^m, \mathbb{A}^{m-1}, \mathbb{A}^{m-2},\dots,\mathbb{A}$, respectively, and $f(0,\dots,0,1)\in\fq$.
\end{lem}
\begin{proof}
We can decompose $P^m$ as the following union of affine spaces: $P^m=\bigcup_{i=0}^m A_i$, where $A_i=\{ Q=[Q_0:\dots:Q_m]\in P^m\mid Q_0=\cdots=Q_{i-1}=0,Q_i=1 \} $ if $1\leq i \leq m$, and $A_0=\{  Q=[Q_0:\dots:Q_m]\in P^m\mid  Q_0=1 \}$. Therefore, $f$ evaluates to $\fq$ in $P^m$ if and only if $f$ evaluates to $\fq$ in each set $A_i$, $0\leq i\leq m$. The evaluation of $\fq$ at each of the points of the set $A_i$ is the same as the evaluation of $f(0,\dots,0,1,x_{i+1},\dots,x_m)$, and the evaluation of this polynomial at the points of $A_i$ is the same as its evaluation in $\mathbb{A}^{m-i}$.
\end{proof}

In order to construct polynomials that evaluate to $\fq$ in $P^m$ we consider homogenizations of traces of polynomials. Given a polynomial $f\in R=\fqs[x_1,\dots,x_m]$, and a degree $d\geq \deg(f)$, we define the homogenization of $f$ up to degree $d$ as $f^h=x_0^df(x_1/x_0,x_2/x_0,\dots,x_m/x_0)\in S_d=\fqs[x_0,\dots,x_m]_d$. In what follows, we will always consider some fixed degree $d$, and, unless stated otherwise, we will assume that we homogenize up to degree $d$.

Let $d\geq 1$ and let $a\in \mathcal{A}_{\leq d}$. We are interested in homogenizing the polynomials from the basis from Theorem \ref{baseafin}. The condition $a\in \mathcal{A}_{\leq d}$ ensures that, with the fixed representatives that we have chosen for $\mathcal{T}_a(f)$ (the exponents of the monomials are reduced modulo $q^s-1$), we have $\deg(\mathcal{T}_a(f))\leq d$. Now we can define the following homogenization: 

\begin{equation}\label{trazahomo}
\mathcal{T}^h_a:R \rightarrow S/I(P^m),\; f\mapsto (\mathcal{T}_a(f))^h,
\end{equation}
where we homogenize up to degree $d$, and we consider that $\mathcal{T}_a(f)\in R$ is the representative that we have chosen in $R_{q^s}$. Note that the homogenization is not well defined in general for a class in $R_{q^s}$, which is why we had to fix a representative for $\mathcal{T}_a(f)$. 

These homogenized traces have already been used to obtain bases for the subfield subcode of a projective Reed-Solomon code and its dual in \cite{sanjoseSSCPRS}. With respect to the dual code of a subfield subcode, we have the following result by Delsarte \cite{delsarte}:

\begin{thm}\label{delsarte}
Let $C\subset \F_{q^s}^n$ be a linear code.
$$
(C\cap \fq^n)^\perp=\Tr(C^\perp),
$$
where $\Tr:\F_{q^s} \rightarrow \fq$ maps $x$ to $x+x^q+\cdots + x^{q^{s-1}}$ and is applied componentwise to $C^\perp$.
\end{thm}

In \cite{sanjoseSSCPRS}, a basis for the dual of the subfield subcode of a projective Reed-Solomon code was obtained by using the previous result. In the following sections we will generalize these ideas to deal with the case $P^m$, with $m>1$.

\section{Codes over the projective plane}\label{secp2}
In this section, we focus on the case $X=P^2$, where we can give precise results, although it gets much more technical than the case $m=1$ from \cite{sanjoseSSCPRS}. The goal is to compute bases for $\PRM_d^{\sigma,\perp}(2)$ and $\PRM_d^\sigma(2)$ and, in particular, their dimensions. We set $S=\fqs[x_0,x_1,x_2]$, and consider cyclotomic sets in two coordinates. Here, $\mathcal{A}$ will be the set of minimal representatives of cyclotomic sets in two coordinates, and we will usually use the letters $a$ and $c$ to denote elements $(a_1,a_2)$ and $(c_1,c_2)$ of some cyclotomic sets $\II_a$ or $\II_c$. We will also use the univariate cyclotomic sets in this context, and we define $\mathcal{A}^1:=\{a_2\mid (a_1,a_2)\in \mathcal{A}\}$. Because of the choice of the ordering of the elements in $\zmsdos$, $a=(a_1,a_2)\in \mathcal{A}$ verifies that $a_2$ is a minimal representative of the cyclotomic set $\II_{a_2}$ in one coordinate. Therefore, $\mathcal{A}^1$ is also the set of minimal representatives of cyclotomic sets in one coordinate. We will use letters $a_2$ or $c_2$ (or a letter that clearly corresponds to an integer) to denote the elements of the cyclotomic sets $\II_{a_2}$ in one coordinate. 

The next result summarizes the main consequences of the results of this section. The definitions of $\overline{d}$ and $Y$ can be found in Definition \ref{defbar} and (\ref{defY}), respectively.

\begin{thm}
Let $1\leq d\leq 2(q^s-1)$. Then the subfield subcode of the projective Reed-Muller code, $\PRM_d^\sigma(2)$, is a code with length $n=\abs{P^m}=\frac{q^{m+1}-1}{q-1}$, and dimension
$$
\dim(\PRM_d^\sigma(2))=\displaystyle \sum_{a\in \mathcal{A}_{<d}} n_a + \sum_{a_2\in Y} n_{a_2}+\epsilon,
$$
where, if we consider $b_2\in \mathcal{A}^1$ with $\II_{b_2}=\II_{\overline{d}}$, then $\epsilon=n_{\overline{d}}+1$ if $\II_{(q^s-1,\overline{d})}\subset \Delta_{\leq d}$; $\epsilon=1$ if $\II_{(q^s-1,\overline{d})}\not \subset \Delta_{\leq d}$ and $\condicionnn$; and $\epsilon=0$ otherwise. Moreover, the minimum distance is bounded by
$$
\wt(\PRM_d^\sigma(2))\geq \wt(\PRM_d(2))=(q^s-t)q^{s(1-r)},
$$
where $d-1=r(q^s-1)+t$, with $0\leq t<q^s-1$.
\end{thm}

The formula for the dimension in the previous result can be found in Corollary \ref{dimprimario}. The dimension of $\PRM_d^{\sigma,\perp}(2)$ can be derived from the previous result, but we also provide another formula in Corollary \ref{dimdual}. Moreover, in Theorem \ref{baseprmprimario} and Theorem \ref{basedualPRM} we provide bases for $\PRM_d^{\sigma}(2)$ and $\PRM_d^{\sigma,\perp}(2)$, respectively.

\subsection{Dual codes of the subfield subcodes of projective Reed-Muller codes}

We start by computing a basis for the dual of the subfield subcode of a projective Reed-Muller code since it is slightly easier due to the nature of Delsarte's Theorem, Theorem \ref{delsarte}. For this we need the following result from \cite{sorensen} about the dual of a projective Reed-Muller code.

\begin{thm}\label{dualPRM}
Let $m\geq 2$, $1\leq d\leq m(q^s-1)$ and $d^\perp=m(q^s-1)-d$. Then
$$
\begin{aligned}
&\PRM_d^\perp(m)=\PRM_{d^\perp}(m) &\text{ for } d\not\equiv 0\bmod (q^s-1), \\
&\PRM_d^\perp(m)=\PRM_{d^\perp}(m)+\langle (1,\dots,1) \rangle &\text{ for } d\equiv 0\bmod (q^s-1).
\end{aligned}
$$
\end{thm}

Setting $m=2$ now, in order to use Delsarte's Theorem \ref{delsarte}, it is useful to introduce the following trace map 
$$
\mathcal{T}:S/I(P^2)\rightarrow S/I(P^2), f\mapsto f+f^q+\cdots+f^{q^{s-1}}.
$$ 
With this definition, it is clear that $\ev \circ \mathcal{T}=\Tr\circ \ev$. Hence, the trace code $ \Tr(\PRM_d^\perp(m))$ can be seen as the code generated by the evaluation of some traces in this case. In particular, we can consider $\T(S_{d^\perp})$ (if $d\equiv 0 \bmod q^s-1$, we also consider $\T(\lambda \cdot 1)$, $\lambda\in \fqs$). The image by the evaluation map of $\T(S_{d^\perp})$ is a system of generators of $ \Tr(\PRM_d^\perp(m))$ if $d\not \equiv 0\bmod q^s-1$. If we extract a maximal linearly independent set of polynomials from $\T(S_{d^\perp})$, then its image by $\ev$ will be a basis for the dual of the subfield subcode.

As we said before, the kernel of the evaluation map is precisely $I(P^2)$, and we have an isomorphism of the primary code with $S/I(P^2)$. The ideal $I(P^2)$ will play a crucial role in understanding linear independence of the polynomials in $\T(S_d)$. Hence, it is helpful to obtain a Gröbner basis for this ideal and a basis for the quotient $S/I(P^2)$. The following result is a consequence of Theorem \ref{vanishingideal} and Lemma \ref{basePm}, which will be proven in Section \ref{secpm}. 
\begin{lem}\label{vanishingidealP2}
The following polynomials form a universal Gröbner basis of $I(P^2)$:
$$
I(P^2)=\langle  x_0^2-x_0,x_1^{q^s}-x_1,x_2^{q^s}-x_2,(x_0-1)(x_1^2-x_1),(x_0-1)(x_1-1)(x_2-1) \rangle.
$$
Moreover, the set of monomials $\{x_1^{a_1}x_2^{a_2}, x_0x_2^{a_2},x_0x_1 \mid 0\leq a_i\leq q^s-1,1\leq i\leq 2\}$ is a basis for $S/I(P^2)$.
\end{lem}

\begin{rem}\label{remequiscero}
Because of the generator $x_0^2-x_0$ of the previous ideal, any positive power of $x_0$ is equivalent to $x_0$ in the quotient ring. Therefore, any monomial $x_0^{a_0}x_1^{a_1}x_2^{a_2}$ with $a_0>0$ is equivalent to $x_0x_1^{a_1}x_2^{a_2}$ in $S/I(P^2)$.
\end{rem}

In what follows, we assume $d\not \equiv 0\bmod q^s-1$ to avoid making exceptions due to Theorem \ref{dualPRM} (we will recover this case later). By Theorem \ref{delsarte} and Theorem \ref{dualPRM}, we have that $\PRM_d^{\sigma,\perp}(2)$ can be generated by the image by the evaluation map of traces (using the map $\T$) of multiples of the monomials of degree $d^\perp$. We show next that, to obtain a basis for the dual code, it is enough to consider the trace maps $\mathcal{T}_a$ instead of $\T$, which we extend from $R_{q^s}$ to $S/I(P^2)$ in the following way:
$$
\mathcal{T}_a:S/I(P^2)\rightarrow S/I(P^2),\:\: f\mapsto f+f^q+\cdots + f^{q^{(n_a-1)}},
$$
for a certain $a\in \mathcal{A}$.

We consider the trace map from $\F_{q^s}$ to $\F_{q^{l}}$, $\Tr_{\F_{q^s}/\F_{q^{l}}}$ (with $l\mid s$): $\Tr_{\F_{q^s}/\F_{q^{l}}}(x)=x+x^{q^{l}}+\cdots +x^{q^{l (\frac{s}{l}-1)}}$. By Theorem \ref{delsarte}, Theorem \ref{dualPRM}, and the previous discussion, we have that $\Tr(\PRM_d^\perp(2))$ is generated by the evaluation of $\T(S_{d^\perp})$, which is generated by the set $\{\mathcal{T}(\lambda x^\gamma), \lambda\in \F_{q^s}^*, x^\gamma \in S_{d^\perp}\}$. Let $\lambda \in \F_{q^s}^*$, $\gamma=(\gamma_0,\gamma_1,\gamma_2)$ and $\hat{\gamma}=(\gamma_1,\gamma_2)$. We consider the cyclotomic set $\II_{\hat{\gamma}}$, and we have that
\begin{equation}\label{trazadelatraza}
\begin{split}
\mathcal{T}(\lambda x^\gamma)\equiv &\lambda x^\gamma + \lambda^q x^{q\gamma}+\cdots+\lambda^{q^{n_{\hat{\gamma}}-1}}x^{q^{n_{\hat{\gamma}}-1}\gamma }\\
&+\lambda^{q^{n_{\hat{\gamma}}}}x^\gamma+\lambda^{q^{n_{\hat{\gamma}}+1}}x^{q\gamma}+\cdots+\lambda^{q^{2n_{\hat{\gamma}}-1}}x^{q^{n_{\hat{\gamma}}-1}\gamma }+\cdots\\
\equiv & \Tr_{\F_{q^s}/\F_{q^{n_{\hat{\gamma}}}}}(\lambda)x^\gamma + (\Tr_{\F_{q^s}/\F_{q^{n_{\hat{\gamma}}}}}(\lambda))^q x^{q\gamma}+\cdots \\
\equiv&\mathcal{T}_{\hat{\gamma}}\left(\Tr_{\F_{q^s}/\F_{q^{n_{\hat{\gamma}}}}}(\lambda)x^\gamma \right) \bmod I(P^2),
\end{split}
\end{equation}
where, if $\gamma_0>0$, we can reduce the exponent of $x_0$ in each monomial to 1 (see Remark \ref{remequiscero}), and we are using that $(x_1^{\gamma_1} x_2^{\gamma_2})^{q^{n_{\hat{\gamma}}}}\equiv x_1^{\gamma_1} x_2^{\gamma_2}\bmod S/I(P^2)$. Equation (\ref{trazadelatraza}) shows that, for each monomial $x^\gamma$, it is enough to consider the traces
\begin{equation}\label{trazasgeneradoras}
\{ \mathcal{T}_{\hat{\gamma}}(\xi_{\hat{\gamma}}^r x^\gamma)\mid 0\leq r\leq n_{\hat{\gamma}}-1\}.
\end{equation}
This is because the trace function is surjective, which means that every element of $\F_{q^{n_{\hat{\gamma}}}}$ is obtained as $\Tr_{\F_{q^s}/\F_{q^{n_{\hat{\gamma}}}}}(\lambda)$ for some $\lambda\in \F_{q^s}$. Taking into account the linearity of the trace function, and the fact that $\{ 1,\xi_{\hat{\gamma}},\dots,\xi_{\hat{\gamma}}^{n_{\hat{\gamma}}-1}\}$ constitutes a basis for $\F_{q^{n_{\hat{\gamma}}}}$, we obtain what we stated. 

Thus, for computing a basis for $\PRM_d^{\sigma,\perp}(2)$, we just need to consider the union of the sets in (\ref{trazasgeneradoras}), and extract a maximal linearly independent set. In principle, we will not see the dual code as the image by the evaluation map of a set of homogeneous polynomials. This makes Lemma \ref{vanishingidealP2} specially valuable in order to argue about linear independence when we consider polynomials of different degree (for homogeneous polynomials, the homogeneous ideal $I(\mathbb{P}^m)$ from \cite{mercier} can be used to discuss linear independence).

We note that, for $d>2(q-1)$, $\PRM_d(2)$ is the whole space. Hence, we will always assume that $d\leq 2(q-1)$ in what follows. We introduce now the following sets which play a crucial role in grouping the polynomials in $S_d$ with linearly dependent traces.

\begin{defn}\label{defbar}
Let $1\leq d \leq 2(q-1)$. For $0\leq b\leq 2(q-1)$, we define $\overline{b}$ as the representative of $b\bmod (q^s-1)$ between 1 and $q^s-1$ if $b\neq 0$, and 0 otherwise. For $a=(a_1,a_2)\in \mathcal{A}$, we define 
$$
M_a(d)=\langle x_0^{b_0}x_1^{b_1}x_2^{b_2}\mid (\overline{b_1},\overline{b_2})\in \II_a, b_0+b_1+b_2=d\rangle \subset S_d.
$$
\end{defn}

It is clear that the union of these sets contains all the monomials of $S_d$, which implies that $S_d=\langle \bigcup_{a\in \mathcal{A}} M_a(d) \rangle$. Therefore, we have that $\mathcal{T}(S_d)=\langle\bigcup_{a\in \mathcal{A}}\mathcal{T}(\langle M_a(d) \rangle)\rangle$, where we have used the linearity of $\T$. Thus, in order to obtain a set of polynomials such that its image by the evaluation map is a basis for $\PRM_d^{\sigma,\perp}(2)$, we are going to obtain a basis for $\mathcal{T}(M_a(d))$, for each $a\in \mathcal{A}$, and then consider the union of these bases which, by the previous argument, will generate $\mathcal{T}(S_d)$. We will then extract a basis from this union.

To achieve that, we first introduce the following definition that we use throughout this section.

\begin{defn}
Let $1\leq d \leq 2(q^s-1)$. We will say that $M_a(d)$ \textit{contains monomials of the two types} if there are monomials $m_1,m_2\in M_a(d)$ such that $x_0\mid m_1$ and $x_0 \nmid m_2$.
\end{defn}

Using all the previous notation, we have the following result which translates some conditions on cyclotomic sets into conditions on the sets $M_a(d)$.

\begin{lem}\label{condicionesM}
Let $1\leq d \leq 2(q^s-1)$. We have the following:
\begin{enumerate}
\item $M_a(d)$ is not empty if and only if $\II_a\cap \Delta_{\leq d}\neq \emptyset$.
\item $x_0$ divides some monomial in $M_a(d)$ if and only if $\II_a\cap \Delta_{<d}\neq \emptyset$.
\item $x_0$ does not divide all the monomials in $M_a(d)$ if and only if $\II_a\cap(\Delta_d\cup\Delta_{\overline{d}})\neq \emptyset$.
\item\label{itemdostipos} $M_a(d)$ contains monomials of the two types if and only if $\II_a\cap \Delta_{< d}\neq \emptyset$ and $\II_a\cap (\Delta_d\cup \Delta_{\overline{d}})\neq \emptyset$. 
\item $x_0$ does not divide any monomial in $M_a(d)\neq \emptyset$ if and only if $\II_a\cap \Delta_{\leq d}\subset \Delta_d$.
\end{enumerate}
\end{lem}
\begin{proof}
The first one is clear from the definitions. We prove (\ref{itemdostipos}) first. If $M_a(d)$ contains monomials of the two types, then $M_a(d)$ is not empty, and there is a monomial $x_1^{b_1}x_2^{b_2}\in M_a(d)$. This means that $(\overline{b_1},\overline{b_2})\in \II_a$, and we have $\overline{b_1}+\overline{b_2}\equiv d\bmod (q^s-1)$. Hence, $(\overline{b_1},\overline{b_2})\in \II_a\cap (\Delta_d\cup \Delta_{\overline{d}})\neq \emptyset$. There is also a monomial $x_0^{c_0}x_1^{c_1}x_2^{c_2}\in M_a(d)$ with $c_0>0$, which implies that $\overline{c_1}+\overline{c_2}<d$ and $(\overline{c_1},\overline{c_2})\in \II_a\cap \Delta_{< d}$.

Conversely, if we have $c\in \II_a$ such that $c_1+c_2\equiv d\bmod(q^s-1)$, this means that, if $c_1>0$, $x_1^{c_1+\lambda (q^s-1)}x_2^{c_2}$ has degree $d$ for some $\lambda\in \{0,1\}$, which means that this monomial would be in $M_a(d)$. If $c_1=0$, the same reasoning proves that the monomial $x_2^{c_2+\lambda (q^s-1)}$ would be in $M_a(d)$ for some $\lambda\in \{0,1\}$. Taking into account the condition $\II_a\cap \Delta_{< d}\neq \emptyset$, there is an element $u\in \II_a$ such that $x_1^{u_1}x_2^{u_2}$ is of degree less than $d$. Thus, $x_0^{u_0}x_1^{u_1}x_2^{u_2}\in M_a(d)$, where $u_0=d-u_1-u_2$. This proves (\ref{itemdostipos}).

By adapting the previous argument, it is easy to prove (2) and (3), and (5) is the negation of (2), taking (1) into account.
\end{proof}

\begin{ex}
We can consider the extension $\F_{16}\supset \F_2$ ($q=2$, $s=4$), and the cyclotomic set $\II_{(0,3)}=\langle(0,3),(0,6),(0,9),(0,12)\rangle$. For $1\leq d\leq 2$ we have that $M_{(0,3)}(d)=\emptyset$ since $\II_{(0,3)}\cap \Delta_{\leq 2}=\emptyset$. For $d=3$, we have $M_{(0,3)}(3)=\langle x_2^3\rangle$, i.e., $x_0$ does not divide any monomial in $M_{(0,3)}(3)$ (due to the fact that $\II_{(0,3)}\cap \Delta_{\leq 3}=\{(0,3)\}\subset \Delta_3$). For $d=5$ (similarly for $d=4$), we have that $M_{(0,3)}(5)=\langle x_0^2x_2^3\rangle$, i.e., $x_0$ divides all the monomials in $M_{(0,3)}(5)$ (precisely because $\II_{(0,3)}\cap \Delta_5=\emptyset$). For $d=6$ we have $M_{(0,3)}(6)=\langle x_0^3x_2^3,x_2^6\rangle$, i.e., $M_{(0,3)}(6)$ contains monomials of the two types, since we have $(0,3)\in \II_{(0,3)}\cap \Delta_{<6}$ and $(0,6)\in \II_a\cap \Delta_6$. Lastly, if we consider a degree higher than $q^s=16$, we have to take into account $\overline{d}$. For example, for $d=18$, we have $M_{(0,3)}(18)=\langle x_0^{15}x_2^3,x_2^{18},x_0^{12}x_2^6,x_0^{9}x_2^9,x_0^6x_2^{12}\rangle$. We see that $M_{(0,3)}(18)$ contains monomials of the two types, as we have that $\overline{d}=3$ and $(0,3)\in \II_{(0,3)}\cap \Delta_{3}$, which means that we can consider the monomial $x_2^{18}\equiv x^3\bmod I(P^2)$, which does not have $x_0$ in its support.
\end{ex}

The following result is a consequence of Lemma \ref{divisionPm}, which is proved in Section \ref{secpm}. It will allow us to obtain a basis for $\T(M_a(d))$, for each $a\in \mathcal{A}$, and it can be understood as the remainder after using the multivariate division algorithm of a monomial with respect to the Gröbner basis from Lemma \ref{vanishingidealP2}. 

\begin{lem}\label{divisionP2}
Let $a_0,a_1,a_2$ be integers, with $a_0,a_1>0$. We have that
$$
\begin{aligned}
x_0^{a_0}x_1^{a_1}x_2^{a_2}&\equiv x_1^{a_1}x_2^{a_2}+x_0x_2^{a_2}-x_2^{a_2}+x_0x_1-x_0-x_1+1 \bmod I(P^2)\\
&\equiv x_1^{a_1}x_2^{a_2}+(x_0-1)(x_2^{a_2}+x_1-1) \bmod I(P^2).
\end{aligned}
$$
\end{lem}

We recall that the kernel of $\ev$ is $I(P^2)$. This implies that a set of classes (polynomials) in $S/I(P^2)$ is linearly independent if and only if the evaluation of this set is linearly independent. This is why, in the following, we may argue about linear independence both from the point of view of polynomials in $S/I(P^2)$ and vectors in $\fqs^n$. 

\begin{lem}\label{lemaindependence}
Let $a=(a_1,a_2)\in \mathcal{A}$, let $\xi_a$ be a primitive element of $\F_{q^{n_a}}$ and let $\xi_{a_2}$ be a primitive element of $\F_{q^{n_{a_2}}}$. Then the following polynomials constitute a basis in $S/I(P^2)$ for $\mathcal{T}(M_a(d))=\langle \mathcal{T}(\lambda x_0^{b_0}x_1^{b_1}x_2^{b_2}),\lambda\in \fqs, x_0^{b_0}x_1^{b_1}x_2^{b_2}\in M_a(d) \rangle$:

\begin{enumerate}
    \item \label{solox0} If $x_0$ divides all the monomials in $M_a(d)\neq\emptyset$:
    $$
    \{ \mathcal{T}_a(\xi_a^r x_0x_1^{a_1}x_2^{a_2})\mid 0\leq r \leq n_a-1 \}.
    $$
    \item \label{solosinx0} If $x_0$ does not divide any monomial in $M_a(d)\neq \emptyset$:
    $$
    \{ \mathcal{T}_a(\xi_a^r x_1^{a_1}x_2^{a_2})\mid 0\leq r \leq n_a-1 \}.
    $$
    \item \label{sinx1} If $M_a(d)$ contains monomials of the two types, and $a_1=0$:
    $$
    \{ \mathcal{T}_a(\xi_a^r x_2^{a_2})\mid 0\leq r \leq n_a-1 \} \cup \{ \mathcal{T}_a(\xi_a^r x_0x_2^{a_2})\mid 0\leq r \leq n_a-1 \}. 
    $$
    \item \label{comp}  If $M_a(d)$ contains monomials of the two types, and $a_1>0$:
    $$
    \{ \mathcal{T}_a(\xi_a^r x_1^{a_1}x_2^{a_2})\mid 0\leq r \leq n_a-1 \} \cup \{(x_0-1)(\mathcal{T}_{a_2}(\xi_{a_2}^rx_2^{a_2})+\mathcal{T}_{a_2}(\xi_{a_2}^r)(x_1-1))\mid 0\leq r\leq n_{a_2}-1\}.
    $$
\end{enumerate}
\end{lem}
\begin{proof}
The fact that the polynomials of each set $\{\mathcal{T}_a(\xi_a^r x_0^{a_0}x_1^{a_1}x_2^{a_2}) , 0\leq r\leq n_a\}$ are linearly independent can easily be seen since the evaluation of each set  $\{\mathcal{T}_a(\xi_a^r x_0^{a_0}x_1^{a_1}x_2^{a_2}) , 0\leq r\leq n_a\}$ in $[\{1\}\times\F_{q^s}^2]$ is the same as the evaluation of  $\{\mathcal{T}_a(\xi_a^r x_1^{a_1}x_2^{a_2}) , 0\leq r\leq n_a\}$ in $\fqs^2$, and we know that the evaluation of this set in $\fqs^2$ is linearly independent it is part of the basis given in Theorem \ref{baseafin} for the affine case. For each monomial $x_0^{b_0}x_1^{b_1}x_2^{b_2}\in M_a(d)$, because of the discussion that led to (\ref{trazasgeneradoras}), we know that, instead of considering the set $\{\mathcal{T}(\lambda x_0^{b_0}x_1^{b_1}x_2^{b_2}),\lambda\in \fqs\}$, it is enough to consider the set $\{\mathcal{T}_a(\xi_a^r x_0^{b_0}x_1^{b_1}x_2^{b_2}) , 0\leq r\leq n_a\}$. 

Therefore, if we consider $x_0^{b_0}x_1^{b_1}x_2^{b_2},x_0^{c_0}x_1^{c_1}x_2^{c_2}\in M_a(d)$, with $b_0,c_0>0$, we know that it is sufficient to consider the traces $\{\mathcal{T}_a(\xi_a^r x_0^{b_0}x_1^{b_1}x_2^{b_2}) , 0\leq r\leq n_a-1\}$ and $\{\mathcal{T}_b(\xi_b^r x_0^{c_0}x_1^{c_1}x_2^{c_2}) , 0\leq r\leq n_a-1\}$ for each monomial, respectively. However, the evaluations of these sets of traces generate the same space in $[\{1\}\times\F_{q^s}^2]$ due to Theorem \ref{baseafin} and Remark \ref{remafin}, and in the rest of the points both sets of polynomials evaluate to 0. For the case with $b_0=c_0=0$, we just need to observe that the evaluation of any polynomial $f(x_1,x_2)$ in $P^2$ is fixed by its evaluation in $[\{1\}\times \F_{q^s}^2]$. By the same argument as before, the evaluations of the two sets of polynomials we are considering in $[\{1\}\times \F_{q^s}^2]$ generate the same space, and by the previous observation this implies that their evaluations over $P^2$ generate the same vector space.

Hence, if we consider the traces of the monomials in $M_a(d)$, it is enough to consider the traces of a monomial divisible by $x_0$ (if any) and the traces of a monomial not divisible by $x_0$ (if any). In fact, we can assume that we are considering the monomials $x_0x_1^{a_1}x_2^{a_2}$ and $x_1^{a_1}x_2^{a_2}$, as any other choice for a monomial that is divisible by $x_0$ and a monomial that is not divisible by $x_0$, respectively, would span the same space when considering the space generated by the traces. In the case where $M_a(d)$ only has monomials of one of those types, we know that those traces are linearly independent and we obtain the cases (\ref{solox0}) and (\ref{solosinx0}). Another easy case is when $a_1=0$, in which, if $M_a(d)$ contains monomials of the two types, we just obtain the polynomials
$$
    \{ \mathcal{T}_a(\xi_a^r x_2^{a_2})\mid 0\leq r \leq n_a-1 \} \cup \{ \mathcal{T}_a(\xi_a^r x_0x_2^{a_2})\mid 0\leq r \leq n_a-1 \}. 
$$
We have seen that both of these sets are linearly independent, and when we consider the union we still keep the linear independence since the monomials of each of these traces are part of the basis in Lemma \ref{vanishingidealP2} and both sets have disjoint support for their polynomials. This corresponds to the case (\ref{sinx1}).

The case where $a_1>0$ and $M_a(d)$ contains monomials of the two types is more involved. By the previous discussion, it is enough to consider the sets of polynomials $\{\mathcal{T}_a(\xi_a^r x_0 x_1^{a_1}x_2^{a_2}) , 0\leq r\leq n_a-1\}$ and $\{\mathcal{T}_a(\xi_a^r x_1^{a_1}x_2^{a_2}) , 0\leq r\leq n_a-1\}$ for generating $\mathcal{T}(M_a(d))$, and we are interested in knowing how many linearly independent polynomials in $S/I(P^2)$ there are in the union of those sets. In order to construct a basis for the space generated by all these polynomials, we start with the polynomials in $\{\mathcal{T}_a(\xi_a^r x_1^{a_1}x_2^{a_2}) , 0\leq r\leq n_a\}$, and we will check which polynomials from the other set can be included without losing linear independence. First, by Lemma \ref{divisionP2}, we have that $$
x_0^{q^l a_0}x_1^{q^l a_1}x_2^{q^l a_2}\equiv x_1^{q^l a_1}x_2^{q^l a_2}+(x_0-1)(x_2^{q^l a_2}+x_1-1)\bmod I(P^2).
$$
Thus, for $a=(a_1,a_2)$ with $a_1>0$, we consider $\II_a$ and $\xi_a$ a primitive element of $\F_{q^{n_a}}$, and we obtain
\begin{equation}\label{trazap2}
\begin{split}
\mathcal{T}_a(\xi_a^r x_0^{a_0}x_1^{a_1}x_2^{a_2})&\equiv \mathcal{T}_a(\xi_a^r x_1^{a_1}x_2^{a_2})+(x_0-1)\sum_{l=0}^{n_a-1}\xi_a^{q^l r}(x_2^{q^l a_2}+x_1-1)\bmod I(P^2)\\
&\equiv \mathcal{T}_a(\xi_a^r x_1^{a_1}x_2^{a_2})+(x_0-1)(\mathcal{T}_a(\xi_a^rx_2^{a_2})+\mathcal{T}_a(\xi_a^r)(x_1-1)) \bmod I(P^2).
\end{split}
\end{equation}
By (\ref{trazap2}), we obtain that we have to see which polynomials of the type
\begin{equation}\label{whichtoadd}
(x_0-1)(\mathcal{T}_a(\xi_a^rx_2^{a_2})+\mathcal{T}_a(\xi_a^r)(x_1-1))=(x_0-1)\mathcal{T}_a(\xi_a^rx_2^{a_2})+(x_0-1)(x_1-1)\mathcal{T}_a(\xi_a^r)
\end{equation}
can be included in the basis that we are constructing without losing linear independence. We note that the exponents of $x_2$ in these polynomials are precisely the elements of $\II_{a_2}$. In fact, these polynomials are closely related to the corresponding traces of $\II_{a_2}$. Arguing as we did to get (\ref{trazadelatraza}), we obtain that
\begin{equation}\label{2trazadelatraza}
\mathcal{T}_a(\xi_a^rx_2^{a_2})=\mathcal{T}_{a_2}\left( \Tr_{\F_{q^{n_a}}/\F_{q^{n_{a_2}}}}(\xi_a^r)x_2^{a_2} \right).
\end{equation}
By the argument we used to get (\ref{trazasgeneradoras}), we see that the set $\{ \mathcal{T}_{a_2}( \xi_{a_2}^r x_2^{a_2} ) \mid 0\leq r\leq n_{a_2}-1\}$, where $\xi_{a_2}$ is a primitive element of $\F_{q^{n_{a_2}}}$, generates the same vector space as $\{ \mathcal{T}_{a}( \xi_a^r x_2^{a_2} ) \mid 0\leq r\leq n_{a}-1\}$. This implies that the set of polynomials 
$$\{(x_0-1)(\mathcal{T}_a(\xi_a^rx_2^{a_2})+\mathcal{T}_a(\xi_a^r)(x_1-1))\mid 0\leq r\leq n_a\}$$
generates the same space as the set 
$$\{(x_0-1)(\mathcal{T}_{a_2}(\xi_{a_2}^rx_2^{a_2})+\mathcal{T}_{a_2}(\xi_{a_2}^r)(x_1-1))\mid 0\leq r\leq n_{a_2}\}.$$
This is because the same linear combination that expresses $\mathcal{T}_a(\xi_a^rx_2^{a_2})$ in terms of the traces $\mathcal{T}_{a_2}( \xi_{a_2}^r x_2^{a_2} )$ also gives $\mathcal{T}_a(\xi_a^r)$ in terms of the traces $ \mathcal{T}_{a_2}(\xi_{a_2}^r)$ (we just evaluate $x_2=1$), and vice versa. Thus, when considering the polynomials from (\ref{whichtoadd}) that we have to include, is is enough to consider
\begin{equation}\label{conjunto2}
\{(x_0-1)(\mathcal{T}_{a_2}(\xi_{a_2}^rx_2^{a_2})+\mathcal{T}_{a_2}(\xi_{a_2}^r)(x_1-1))\mid 0\leq r\leq n_{a_2}-1\},
\end{equation}
which are linearly independent since they coincide with the univariate affine case from Theorem \ref{baseafin} when we evaluate in the points of $[\{0\}\times \{1\}\times\F_{q^s}]$. Finally, when we consider the union of those polynomials with the set $\{\mathcal{T}_a(\xi_a^r x_1^{a_1}x_2^{a_2}) , 0\leq r\leq n_a\}$, we see that they are linearly independent because the polynomials from (\ref{conjunto2}) evaluate to the zero vector in $[\{1\}\times\fqs^2]$, while the polynomials from the set $\{\mathcal{T}_a(\xi_a^r x_1^{a_1}x_2^{a_2}) , 0\leq r\leq n_a\}$ give linearly independent vectors when evaluating in $[\{1\}\times \F_{q^s}^2]$.
\end{proof}

By Lemma \ref{lemaindependence}, if $x_0$ divides all the monomials from $M_a(d)$, or does not divide any of the monomials in $M_a(d)$, we only have to consider $n_a$ polynomials for each $a\in \mathcal{A}$ to construct a basis for $\mathcal{T}(M_a(d))$. However, if $M_a(d)$ contains monomials of the two types we have to consider $n_a+n_{a_2}$ polynomials (note that for $a_1=0$ we have $a=(0,a_2)$ and $2n_a=2n_{a_2}=n_a+n_{a_2}$).

\begin{rem}\label{remcerod}
We note that if $a_1=0$ and $M_a(d)$ contains monomials of the two types, this means that $x_2^d\in M_a(d)$, which implies that $\overline{d}\in \II_{a_2}$. Therefore, the case (\ref{sinx1}) in Lemma \ref{lemaindependence} applies only to $(0,\overline{d})\in \mathcal{A}$, and only when $M_{(0,\overline{d})}$ contains monomials of the two types.
\end{rem}

Let $d^\perp=2(q-1)-d$. We introduce the following notation to state the main result of this section. For each $a\in \mathcal{A}$ such that $M_a(d^\perp)\neq \emptyset$, let $\xi_a$ be a primitive element in $\F_{q^{n_a}}$, and consider the following set:

\begin{itemize}
\item[(a)] If $x_0$ divides all the monomials from $M_a(d^\perp)$, we set
$$
T_a=\{\mathcal{T}_a(\xi_a^r x_0x_1^{a_1}x_2^{a_2})\mid 0\leq r \leq n_a-1\}.
$$
\item[(b)] We set 
$$
T_a=\{\mathcal{T}_a(\xi_a^r x_1^{a_1}x_2^{a_2})\mid 0\leq r \leq n_a-1\}
$$
otherwise.
\end{itemize}

The reasoning behind $T_a$ is that for any $a\in \mathcal{A}$ such that $M_a(d^\perp)\neq \emptyset$, from Lemma \ref{lemaindependence} we obtain that $T_a$ is a set of linearly independent polynomials in $\mathcal{T}(M_a(d^\perp))$. We define $U=\{ a\in \mathcal{A}\mid M_a(d^\perp) \neq \emptyset \}$, and we consider the union of the previous sets:
$$
D_1=\bigcup_{a\in U}T_a.
$$
This is one of the sets that we will consider for constructing a basis for $\PRM_d^{\sigma,\perp}(2)$. 

If $M_a(d^\perp)$ contains monomials of the two types, then, besides $T_a$, Lemma \ref{lemaindependence} states that there are more linearly independent polynomials in $\mathcal{T}(M_a(d^\perp))$. Thus, we turn our attention now to the case (\ref{comp}) from Lemma \ref{lemaindependence}. For each $a_2\in \mathcal{A}^1$, let $\xi_{a_2}$ be a primitive element in $\F_{q^{n_{a_2}}}$, and we consider the set
$$
T_{a_2}=\{(x_0-1)(\mathcal{T}_{a_2}(\xi_{a_2}^rx_2^{a_2})+\mathcal{T}_{a_2}(\xi_{a_2}^r)(x_1-1))\mid 0\leq r\leq n_{a_2}-1\}.
$$
Let $V=\{a_2\in \mathcal{A}^1\mid \II_{a_2}\neq \II_{\overline{d^\perp}} \text{ and } \exists \; c\in\mathcal{A} \text{ with } c_2=a_2 \text{ and } M_c(d^\perp) \text{ contains monomials }$ $ \text {of the two types} \}$, and we consider the set
$$
D_2=\bigcup_{a_2\in V}T_{a_2}.
$$
If we want to generate all the polynomials in $\bigcup_{a\in \mathcal{A}}\mathcal{T}(M_a(d^\perp))$, from Lemma \ref{lemaindependence} we see that we still have to consider the polynomials corresponding to $a\in \mathcal{A}$ such that $\II_{a_2}=\II_{\overline{d}}$. Let us define a set $D_3$ that will contain the polynomials corresponding to this case and that we will consider for constructing a basis for $\PRM_d^{\sigma,\perp}(2)$. Let $a_2\in \mathcal{A}^1$ such that $\II_{a_2}=\II_{\overline{d^\perp}}$, and $\xi_{a_2}$ a primitive element in $\F_{q^{n_{a_2}}}$.
\begin{enumerate}
    \item[(a)]\label{casoa} If $M_{(0,\overline{d^\perp})}(d^\perp)=M_{(0,a_2)}(d^\perp)$ contains monomials of the two types:
    \begin{enumerate}
        \item[(a.1)]\label{casoa1} If there is an element $c\in\mathcal{A}$ such that $c_2=a_2$, $\II_c\neq \II_{(0,\overline{d^\perp})}$, and $M_c(d^\perp)$ contains monomials of the two types, we set
        $$
        D_3=\{ \mathcal{T}_{a_2}(\xi_{a_2}^r x_0x_2^{a_2})\mid 0\leq r \leq n_{a_2}-1 \}\cup \{(x_0-1)(x_1-1)\}.
        $$
        \item[(a.2)]\label{casoa2} We set 
        $$       D_3=\{ \mathcal{T}_{a_2}(\xi_{a_2}^r x_0x_2^{a_2})\mid 0\leq r \leq n_{a_2}-1 \}.
        $$
        otherwise.
    \end{enumerate}
    \item[(b)]\label{casob} We set 
        $$
        D_3=\emptyset
        $$
    otherwise.
\end{enumerate}

We note that the case (b) happens if and only if $x_0$ does not divide any monomial in $M_{(0,\overline{d^\perp})}(d^\perp)$. The precise reason why we define $D_3$ in this way will be clear in the proof of Theorem \ref{basedualPRM}, which we will state after defining one last set, which we are considering just to cover the case in which $d\equiv 0 \bmod q^s-1$. In that case, we also have the evaluation of $1$ in the dual code of $\PRM_d(2)$ by Theorem \ref{dualPRM}. If $d=q^s-1$, we define $D_4=\{1\}$, and $D_4=\emptyset$ otherwise.

\begin{thm}\label{basedualPRM}
Let $d\geq 1$ and $d^\perp=2(q^s-1)-d$. For each $a\in \mathcal{A}$, let $\xi_a$ be a primitive element in $\F_{q^{n_a}}$ such that $\T_{a}(\xi_{a})\neq 0$, and for each $a_2\in \mathcal{A}^1$, let $\xi_{a_2}$ be a primitive element in $\F_{q^{n_{a_2}}}$ such that $\T_{a_2}(\xi_{a_2})\neq 0$ (one can always assume this \cite{cohen}). 
Using the previous definitions, we consider the set
$$
D=D_1\cup D_2\cup D_3\cup D_4.
$$
Then we have that the image by the evaluation map of $D$ forms a basis for $\PRM_d^{\sigma,\perp}(2)$.
\end{thm}

\begin{proof}
Firstly, by Theorem \ref{dualPRM} we know that $\PRM_d^\perp(2)$ is equal to $\PRM_{d^\perp}(2)$, except when $d\equiv 0\bmod (q^s-1)$, in which case we also have to consider the evaluation of the constant 1. If $d\not \equiv 0\bmod (q^s-1)$, by Delsarte's Theorem, Theorem \ref{delsarte}, $\PRM_d^{\sigma,\perp}(2)=\Tr(\PRM_{d^\perp}(2))$, and due to the fact that we have $\Tr\circ \ev = \ev\circ \mathcal{T}$, we see that if we consider $\mathcal{T}(S_{d^\perp})$ (and possibly the constant 1), we obtain a system of generators for $\PRM_d^{\sigma,\perp}(2)$. Therefore, in order to obtain a basis, we just need to study linear independence between these polynomials. In fact, we have $S_{d^\perp}=\langle \bigcup_{a\in\mathcal{A}} M_a(d^\perp) \rangle $, which means that we can consider the union of the bases given for each $\mathcal{T}(M_a(d^\perp))$ from Lemma \ref{lemaindependence}, and we can obtain obtain a basis for $\PRM_d^{\sigma,\perp}(2)$ by extracting a maximal linearly independent set. We focus first on computing a basis for $\mathcal{T}(S_{d^\perp})$, and we will consider the cases where $d\equiv 0\bmod q^s-1$ later. 

In what follows, for each $a\in \mathcal{A}$ we consider $\xi_a$ a primitive element in $\F_{q^{n_a}}$, and for each $a_2\in \mathcal{A}^1$ we consider $\xi_{a_2}$ a primitive element in $\F_{q^{n_{a_2}}}$. By construction, it is clear that we have $D_1\cup D_2\subset \mathcal{T}(S_{d^\perp})$. We show now that also $D_3$ is contained in $\mathcal{T}(S_{d^\perp})$, and $D_4$ is contained in $\mathcal{T}(S_{d^\perp}+\langle 1 \rangle )$ when $D_4\neq \emptyset$.

Let $a_2\in \mathcal{A}^1$ such that $\II_{a_2}=\II_{\overline{d^\perp}}$. For $D_3$, we have to justify that, if $M_{(0,\overline{d^\perp})}(d^\perp)$ contains monomials of the two types and there is an element $c\in \mathcal{A}$ such that $c_2=a_2$, $\II_c\neq \II_{(0,\overline{d^\perp})}$ and $M_c(d^\perp)$ contains monomials of the two types, then $(x_0-1)(x_1-1)$ is in $\T(S_{d^\perp})$. Under these assumptions, by Lemma \ref{lemaindependence} we have that the following sets are in $\T(S_{d^\perp})$:
\begin{equation}\label{polpeso1}
\begin{aligned}
&\{ \mathcal{T}_{(0,a_2)}(\xi_{(0,a_2)}^r x_2^{a_2})\mid 0\leq r \leq n_{(0,a_2)}-1 \} \cup \{ \mathcal{T}_{(0,a_2)}(\xi_{(0,a_2)}^r x_0x_2^{a_2})\mid 0\leq r \leq n_{(0,a_2)}-1 \},\\
&\{ \mathcal{T}_c(\xi_a^r x_1^{c_1}x_2^{c_2})\mid 0\leq r \leq n_c-1 \} \cup \{(x_0-1)(\mathcal{T}_{c_2}(\xi_{c_2}^rx_2^{c_2})+\mathcal{T}_{c_2}(\xi_{c_2}^r)(x_1-1))\mid 0\leq r\leq n_{c_2}-1\}.
\end{aligned}
\end{equation}
Taking into account that $c_2=a_2$, if we assume that $\xi_{(0,a_2)}=\xi_{a_2}$ (note that $n_{a_2}=n_{(0,a_2)}$), then $\mathcal{T}_{(0,a_2)}(\xi_{(0,a_2)}^r x_2^{a_2})=\mathcal{T}_{c_2}(\xi_{c_2}^rx_2^{c_2})$ and  $\mathcal{T}_{(0,a_2)}(\xi_{(0,a_2)}^r x_0x_2^{a_2})=x_0\mathcal{T}_{c_2}(\xi_{c_2}^rx_2^{c_2})$. By assumption, we have that $\T_{c_2}(\xi_{c_2})\neq 0$. Hence, taking into account that we can generate the polynomial $(1-x_0)\mathcal{T}_{c_2}(\xi_{c_2}x_2^{c_2})$ with the first union of sets in (\ref{polpeso1}), we see that with the first union of sets and the last set from (\ref{polpeso1}) we can generate $(x_0-1)(x_1-1)$. Thus, $D_1\cup D_2\cup D_3\subset \T(S_{d^\perp})$. On the other hand, if $d=q^s-1$, we have $D_4=\{1\}$, and it is clear that $D_4\subset \T(S_{d^\perp}+\langle 1\rangle )$. Therefore, we have seen that the image by the evaluation map of $D$ is always in $\PRM_d^{\sigma,\perp}(2)$.

Now we justify that the evaluation of the polynomials in $D$ is linearly independent. If we consider the monomials $x_0^{a_0}x_1^{a_1}x_2^{a_2}$, $x_0^{b_0}x_1^{b_1}x_2^{b_2}$, of degree $d^\perp$, with $\II_a\neq \II_b$ (for $a=(a_1,a_2)$, $b=(b_1,b_2)$), then we have that the sets $\{\mathcal{T}_a(\xi_a^r x_0^{a_0}x_1^{a_1}x_2^{a_2}) , 0\leq r\leq n_a-1\}$ and $\{\mathcal{T}_b(\xi_b^r x_0^{b_0}x_1^{b_1}x_2^{b_2}) , 0\leq r\leq n_b-1\}$ are linearly independent since in $[\{1\}\times\F_{q^s}^2]$ they are linearly independent by the affine case from Theorem \ref{baseafin} in two variables. Using Lemma \ref{lemaindependence} we see that the polynomials in $D_1$ are linearly independent. 

Each polynomial $(x_0-1)(\mathcal{T}_{a_2}(\xi_{a_2}^rx_2^{a_2})+\mathcal{T}_{a_2}(\xi_{a_2}^r)(x_1-1))$, with $0\leq r\leq n_{a_2}-1$, has the same evaluation in $[\{0\}\times \{1\}\times\fqs]$ as $-\mathcal{T}_{a_2}(\xi_{a_2}^rx_2^{a_2})$ in $\fqs$. Hence, the evaluation of the polynomials in $D_2$ is linearly independent by Theorem \ref{baseafin} in one variable. Moreover, these polynomials evaluate to $0$ in $[\{1\}\times\fqs^2]$, while the polynomials from $D_1$ have linearly independent evaluation in $[\{1\}\times\fqs^2]$, which means that the evaluation of $D_1\cup D_2$ is also linearly independent.

We show now that a similar reasoning proves that the evaluation of $D_1\cup D_2\cup D_3$ is also linearly independent. Looking at the definition of $D_3$, if we are in the case (a.1), the evaluation of the polynomial $(x_0-1)(x_1-1)$ is linearly independent from the evaluation of the rest of polynomials in $D_1\cup D_2\cup D_3$ as it is the only one that evaluates to $0$ in $[\{1\}\times\fqs^2]$ and $[\{0\}\times \{1\}\times\fqs]$, and the rest of polynomials have linearly independent evaluations in those sets. Let $a_2\in \mathcal{A}^1$ such that $\II_{a_2}=\II_{\overline{d^\perp}}$. The evaluation of $\mathcal{T}_{a_2}(\xi_{a_2}^r x_0x_2^{a_2})$, for some $0\leq r \leq n_{a_2}-1$, is linearly independent from the evaluation of any polynomial in $D_1$, besides $\mathcal{T}_{a_2}(\xi_{a_2}^r x_2^{a_2})$, due to the argument we used to discuss linear independence between elements in $D_1$. But its evaluation is also linearly independent from the evaluation of $\mathcal{T}_{a_2}(\xi_{a_2}^r x_2^{a_2})$ by Lemma \ref{lemaindependence} (\ref{sinx1}). The same argument that we used to prove that the evaluation of the polynomials in $D_1\cup D_2$ is linearly independent proves that the evaluation of $\mathcal{T}_{a_2}(\xi_{a_2}^r x_0x_2^{a_2})$ is linearly independent from the evaluation of the polynomials in $D_2$. Thus, in this case, the evaluation of $D_1\cup D_2\cup D_3$ is linearly independent. The same arguments prove that $D_1\cup D_2\cup D_3$ is linearly independent in the other cases that appear in the definition of $D_3$. 

We study now the cases in which we have $D_4\neq \emptyset$, i.e., the case where $d=q^s-1$. The evaluation of the constant 1 is linearly independent from the evaluation of the rest of polynomials in this case since, if we look at the evaluation in $[\{0\}\times \{1\}\times\fqs]$, the constant 1 is linearly independent from the evaluation of the rest of univariate traces by Theorem \ref{baseafin}. Hence if we had a linear combination of polynomials from $D_1\cup D_2\cup D_3$ with the same evaluation as 1 in $P^2$, when setting $x_0=0,x_1=1$, the result would be the constant 1.
If we look at the polynomials that we have in $D_1\cup D_2\cup D_3$, the only polynomial that would have a constant in its support after setting $x_0=0,x_1=1,$ would be the only polynomial in $T_{0}$: $(x_0-1)(1+(x_1-1))=(x_0-1)x_1$. However, we only consider this polynomial in $D_2$ if there is some $b\in\mathcal{A}$ such that $\II_{b_2}=\II_{0}=\{0\}$ and if $M_b(d^\perp)=M_b(q^s-1)$ contains monomials of the two types. Therefore, $b_2=0$, and we must have $b_1=q^s-1$ if we want to have some monomial that is not divided by $x_0$ in $M_b(q^s-1)$ by Lemma \ref{condicionesM}. However, $M_{(q^s-1,0)}(q^s-1)=\{x_1^{q^s-1}\}$ does not have monomials of the two types. Thus, the polynomial $(x_0-1)x_1$ is not in $D_1\cup D_2\cup D_3$ in this case and the evaluation of $D=D_1\cup D_2\cup D_3\cup D_4$ is linearly independent.

The only thing left to prove for asserting that $D$ is a basis is that this set is a maximal linearly independent set, or, equivalently, that $D$ generates $\mathcal{T}(S_{d^\perp})$ if $d\not\equiv 0 \bmod q^s-1$, and $D$ generates $\mathcal{T}(S_{d^\perp}+\langle 1\rangle )$ otherwise. To see that $D$ generates $\mathcal{T}(S_{d^\perp})$ when $d\not\equiv 0 \bmod q^s-1$, we have seen that it is enough to check that we can generate all the bases for the sets $\mathcal{T}(M_a(d^\perp))$ from Lemma \ref{lemaindependence}. Let $a\in \mathcal{A}$ such that $M_a(d^\perp)\neq \emptyset$.  If $M_a(d^\perp)$ does not have monomials of the two types, then we see that the basis for $\T(M_a(d^\perp))$ from Lemma \ref{lemaindependence} is contained in $D_1$. If $M_a(d^\perp)$ contains monomials of the two types, then we are in case (\ref{sinx1}) or case (\ref{comp}) from Lemma \ref{lemaindependence}.

Due to the ordering of the elements in $\zmsdos$, $a\in \mathcal{A}$ implies that $a_2\in \mathcal{A}^1$. We consider now the case (\ref{comp}) and we assume first that $\II_{a_2}\neq \II_{\overline{d^\perp}}$. In this situation, it is clear by the definitions that the basis for $\mathcal{T}(M_a(d))$ from Lemma \ref{lemaindependence} is contained in $D_1\cup D_2$.

Now we study the case (\ref{sinx1}) from Lemma \ref{lemaindependence}, and also the case (\ref{comp}) when $\II_{a_2}=\II_{\overline{d^\perp}}$, which are the only cases left. By Remark \ref{remcerod}, in both situations we have that $\II_{a_2}=\II_{\overline{d^\perp}}$. Instead of studying the sets $\mathcal{T}(M_c(d^\perp))$, with $c\in \mathcal{A}$ and $c_2=a_2$, one by one, we consider them together in this case, and we will see that we can generate $\bigcup_{c\in \mathcal{A}\mid c_2=a_2}\mathcal{T}(M_c(d^\perp))$. For each $c\in \mathcal{A}$ with $c_2=a_2$ and $M_c(d^\perp)\neq \emptyset$, if $M_c(d^\perp)$ does not have monomials of the two types, we have already seen that the basis for $\mathcal{T}(M_c(d^\perp))$ from Lemma \ref{lemaindependence} is contained in $D_1$. And if $M_c(d^\perp)$ contains monomials of the two types, then it is also clear that the first set of polynomials that appears in cases (\ref{sinx1}) and (\ref{comp}) from Lemma \ref{lemaindependence} is contained in $D_1$. Thus, we focus on the second set of polynomials from those cases in Lemma \ref{lemaindependence}.

If $M_{(0,\overline{d^\perp})}(d^\perp)=M_{(0,a_2)}(d^\perp)$ contains monomials of the two types, by the definition of $D_3$ we have that the basis for $\mathcal{T}(M_{(0,a_2)}(d^\perp))$ from Lemma \ref{lemaindependence} is contained in $D_1\cup D_3$. If we also have some $c\in \mathcal{A}$, $\II_c\neq \II_{(0,a_2)}$, with $c_2=a_2$, and such that $M_c(d^\perp)$ contains monomials of the two types, then we have that $(x_0-1)(x_1-1)\in D_3$, and by the reasoning that we did after (\ref{polpeso1}) it is clear that we can generate the basis of $\T(M_c(d^\perp))$ given in Lemma \ref{lemaindependence} with the polynomials in $D_1\cup D_2\cup D_3$.
 
If $M_{(0,a_2)}(d^\perp)$ does not have monomials of the two types, we clearly have the basis from Lemma \ref{lemaindependence} for $\mathcal{T}(M_{(0,a_2)}(d^\perp))$ contained in $D_1\cup D_3$. We also note that, by Lemma \ref{condicionesM}, $M_{(0,a_2)}(d^\perp)$ does not have monomials of the two types if and only if $d^\perp=a_2$, i.e., $d^\perp$ is the minimal element in $\II_{d^\perp}$. Hence, for any $c\in\mathcal{A}$ with $c_2=a_2=d^\perp$, $\II_c\neq \II_{(0,d^\perp)}$, we obtain that, for each $\gamma\in \II_c$, we have $\gamma_1\neq 0$ and  $\gamma_1+\gamma_2>c_2=d^\perp$, which means that $M_c(d^\perp)=\emptyset$.

Finally, we have to consider the cases where $d\equiv 0\bmod q^s-1$. If $d=q^s-1$, we already have $1\in D_4$. For the case $d=2(q^s-1)$, we have $\mathcal{T}_{(0,0)}(x_1^0x_2^0)=1$ in $D_1$, which means that we also have the evaluation of the constant 1 when evaluating the polynomials in $D$. Therefore, we have proved that the image by the evaluation map of $D$ is a basis for $\PRM_d^{\sigma,\perp}(2)$. 
\end{proof}

\begin{cor}\label{dimdual}
Let $d\geq 1$ and $d^\perp=2(q^s-1)-d$. Let $U=\{ a\in \mathcal{A}\mid M_a(d^\perp) \neq \emptyset \}$ and $V=\{a_2\in \mathcal{A}^1\mid \II_{a_2}\neq \II_{\overline{d^\perp}} \text{ and } \exists \; c\in\mathcal{A} \text{ with } c_2=a_2 \text{ and } M_c(d^\perp) \text{ contains monomials }$ $ \text {of the two types} \}$ as before. The dimension of $\PRM_d^{\sigma,\perp}(2)$ is
$$
\dim(\PRM_d^{\sigma,\perp}(2))=\abs{D}=\abs{D_1}+\abs{D_2}+\abs{D_3}+\abs{D_4}=\
\sum_{a\in U}n_a+\sum_{a_2\in V}n_{a_2}+\epsilon_3+\epsilon_4,
$$
where $\epsilon_3=n_{\overline{d^\perp}}+1$ if $M_{(0,\overline{d^\perp})}(d^\perp)$ contains monomials of the two types and there is  $\II_c\neq \II_{(0,\overline{d^\perp})}\text{ with } c_2\in \II_{\overline{d}}$ such that $M_c(d^\perp)$ contains monomials of the two types; $\epsilon_3=n_{\overline{d^\perp}}$ if $M_{(0,\overline{d^\perp})}(d^\perp)$ contains monomials of the two types but there is no $\II_c\neq \II_{(0,\overline{d^\perp})}$ as before; and $\epsilon_3=0$ otherwise. Finally, $\epsilon_4=\abs{D_4}$, i.e., $\epsilon_4=1$ if $d=q^s-1$, and $\epsilon_4=0$ otherwise.
\end{cor}

\begin{ex}
Consider the extension $\F_4\supset \F_2$ and let us compute the set $D$ for $d=4$. We have $d^\perp=2$ and, from Example \ref{exciclo}, the set of minimal representatives is $\mathcal{A}=\{(0,0),(1,0),(0,1),(1,1),(3,0),(0,3),(3,3),(2,1),(1,3),(3,1) \}$. We start by constructing the set $D_1$. We consider the minimal representatives $a$ such that $M_a(d^\perp)\neq \emptyset$, which by Lemma \ref{condicionesM} is equivalent to having $\II_a\cap \Delta_{\leq d^\perp}\neq \emptyset$. The only cyclotomic sets that satisfy that condition in this case are $\II_{(0,0)}$, $\II_{(1,0)}$, $\II_{(0,1)}$ and $\II_{(1,1)}$. Therefore, we have $U=\{(0,0),(1,0),(0,1),(1,1)\}$ and $D_1=\bigcup_{a\in U}T_a$. For example, assuming $\xi_{(1,0)}$ is a primitive element of $\F_{4}$, for $a=(1,0)$ we have 
$$
T_{(1,0)}=\{ \T_{(1,0)}(\xi_{(1,0)}^r x_1)\mid 0 \leq r\leq 1\}= \{ \xi_{(1,0)}^r x_1+ \xi^{2r}x_1^2 \mid 0 \leq r\leq 1\}.
$$
We also have $\abs{D_1}=\sum_{a\in U}n_a=7$. For $\abs{D_2}$, we consider $\mathcal{A}^1=\{0,1,3\}$. The only $a\in \mathcal{A}$ such that $M_a(d^\perp)$ contains monomials of the two types are the ones such that $\II_a\cap \Delta_{<d^\perp}\neq \emptyset $ and $\II_a\cap (\Delta_{d^\perp} \cup \Delta_{\overline{d^\perp}})\neq \emptyset$, according to Lemma \ref{condicionesM}. This is a subset of $U$, and from the elements of $U$, the ones that satisfy this condition are $(1,0)$ and $(0,1)$. For example, $\II_{(1,0)}\cap \Delta_{< 2}=(1,0)$ and $\II_{(1,0)}\cap \Delta_2=(2,0)$. Hence, looking at the second coordinate of $(1,0)$ and $(0,1)$, we have $V=\{0,1\}$, and $D_2=\bigcup_{a_2\in V}T_{a_2}$. For example, if we consider $\xi_1=\xi_{(1,0)}$ a primitive element of $\F_4$, for $a_2=1$ we have 
$$
\begin{aligned}
T_{1}=&\{(x_0-1)(\mathcal{T}_{1}(\xi_{1}^rx_2)+\mathcal{T}_{1}(\xi_{1}^r)(x_1-1))\mid 0\leq r\leq 1\}\\
=&\{(x_0-1)(\xi_1^rx_2 +\xi_1^{2r}x_2^2+(\xi_1^r+\xi_1^{2r})(x_1-1))\mid 0\leq r\leq 1\}.
\end{aligned}
$$
We have $\abs{D_2}=\sum_{a_2\in V}n_{a_2}=3$. One can check that $D_3=D_4=\emptyset$ in this case. Thus, the evaluation of the set $D_1\cup D_2$ is a basis for $\PRM_4^{\sigma,\perp}(2)$, and $\dim \PRM_4^{\sigma,\perp}(2)=10$.  
\end{ex}

\subsection{Subfield subcodes of projective Reed-Muller codes}
In this section we compute a basis for $\PRM_d^\sigma(2)$. The discussion gets more technical than in the previous case, but we can obtain explicit results as well. We start by considering some sets of polynomials that we use to construct a basis for the subfield subcode. We recall the notation $\mathcal{A}_{\leq d}=\{a\in\mathcal{A}\mid \II_a\subset \Delta_{\leq d}\}$ and $\mathcal{A}_{< d}=\{a\in\mathcal{A}\mid \II_a\subset \Delta_{< d}\}$. We also consider $\mathcal{A}^1_{\leq d}=\{a_2\in\mathcal{A}^1\mid \forall c_2\in \II_{a_2}, c_2\leq d\}$ for the univariate case. It is also important to recall the definition of homogenized trace from (\ref{trazahomo}). 

\begin{lem}\label{lemab1}
Let $1\leq d\leq 2(q^s-1)$ and let $\xi_{a}$ be a primitive element in $\F_{q^{n_a}}$. The image by the evaluation map of the polynomials in the set
$$
B_1=\bigcup_{a\in \mathcal{A}_{<d}} \{ x_0\mathcal{T}_a(\xi^r_a x_1^{a_1}x_2^{a_2})\mid 0\leq r\leq n_a-1  \},
$$
is in $\PRM_d^\sigma(2)$. Moreover, the evaluation of the polynomials in $B_1$ is linearly independent.
\end{lem}
\begin{proof}

The evaluation of these polynomials in $[\{1\}\times\F_{q^s}^2]$ is the same as the evaluation of the polynomials of the set
$$\bigcup_{a\in \mathcal{A}_{<d}} \{ \mathcal{T}_a(\xi^r_a x_1^{a_1}x_2^{a_2})\mid 0\leq r\leq n_a-1  \}
$$
in $\F_{q^s}^2$. This set of polynomials evaluates to $\fq$ by Theorem \ref{baseafin}, which means that the polynomials in $B_1$ evaluate to $\fq$ in $[\{1\}\times\F_{q^s}^2]$, and they clearly evaluate to $0$ in the rest of the points in $P^2$. By Lemma \ref{lemassc}, each of these polynomials evaluates to $\fq$. We have to see that these polynomials are equivalent modulo $S/I(P^2)$ to some homogeneous polynomials of degree $d$, because in that case these polynomials would have the same evaluation as some homogeneous polynomials of degree $d$, which means that their evaluation is in $\PRM_d^\sigma(2)$. Let $a\in \mathcal{A}_{<d}$. For $0\leq r\leq n_a-1$, we consider the polynomial $\mathcal{T}^h_a(\xi_a^rx_1^{a_1}x_2^{a_2})$, where we homogenize up to degree $d$. Having $a\in \mathcal{A}_{<d}$ means that, after reducing the exponents modulo $q^s-1$, the monomials $x_1^{c_1}x_2^{c_2}$ that appear in the support of $\mathcal{T}_a(\xi_a^rx_1^{a_1}x_2^{a_2})$ satisfy that $c_1+c_2<d$ (these exponents are precisely the elements of $\II_a\subset \Delta_{<d}$). Therefore, after homogenizing up to degree $d$, $x_0$ divides all the monomials in the support of $\mathcal{T}^h_a(\xi_a^rx_1^{a_1}x_2^{a_2})$. Taking into account the equation $x_0^2-x_0$ from $I(P^2)$, this means that $\mathcal{T}^h_a(\xi_a^rx_1^{a_1}x_2^{a_2})\equiv x_0\mathcal{T}_a(\xi_a^rx_1^{a_1}x_2^{a_2})\bmod I(P^2)$ in this case. Hence, the evaluation of the polynomials in $B_1$ is in $\PRM_d^\sigma(2)$.

We finish the proof by noting that their evaluation is linearly independent precisely since their evaluation in $[\{1\}\times\F_{q^s}^2]$ is linearly independent by Theorem \ref{baseafin}.
\end{proof}

\begin{ex}\label{ejemploprimario1}
We consider an extension $\F_{16}\supset \F_{2}$ (i.e., $q=2$, $s=4$), and the goal of the examples in this section is to compute a basis for $\PRM^\sigma_{21}(2)$. We start by computing the set $B_1$, which is a set of linearly independent polynomials that evaluate to $\fq$ by the previous discussion. First of all, we need to consider all the cyclotomic sets $\II_a$ such that $\II_a\subset \Delta_{<21}$. For each of those cyclotomic sets, we consider the corresponding set of traces from $B_1$. For example, we can consider the cyclotomic set $\II_{(1,1)}=\{(1,1),(2,2),(4,4),(8,8) \}$, which gives us the following $n_{(1,1)}=4$ polynomials ($\xi$ is a primitive element in $\F_{2^4}$):
$$
\begin{aligned}
\mathcal{T}_{(1,1)}^h(x_1x_2)&=x_0^{19}x_1x_2+x_0^{17}x_1^2x_2^2+x_0^{13}x_1^4x_2^4+x_0^5x_1^8x_2^8,\\
\mathcal{T}_{(1,1)}^h(\xi x_1x_2)&=\xi x_0^{19}x_1x_2+\xi^2x_0^{17}x_1^2x_2^2+\xi^4x_0^{13}x_1^4x_2^4+\xi^8 x_0^5x_1^8x_2^8, \\
\mathcal{T}_{(1,1)}^h(\xi^2 x_1x_2)&=\xi^2 x_0^{19}x_1x_2+\xi^4x_0^{17}x_1^2x_2^2+\xi^8x_0^{13}x_1^4x_2^4+\xi x_0^5x_1^8x_2^8, \\
\mathcal{T}_{(1,1)}^h(\xi^3 x_1x_2)&=\xi^3 x_0^{19}x_1x_2+\xi^6x_0^{17}x_1^2x_2^2+\xi^{12}x_0^{13}x_1^4x_2^4+\xi^9 x_0^5x_1^8x_2^8,
\end{aligned}
$$
where we see that we are homogenizing up to degree $d=21$. As we have said in the previous discussion, these polynomials are linearly independent because in $[\{1\}\times\fqs^2]$ they have the same evaluation as the traces $\mathcal{T}_{(1,1)}(\xi^r x_1x_2)$, $0\leq r\leq n_{(1,1)}-1$, that would appear in the affine case from Theorem \ref{baseafin}. And they clearly evaluate to $\fq$, as they evaluate to 0 in the rest of the points of $P^2$. We can continue doing this for all the other cyclotomic sets such that $\II_a\subset\Delta_{<21}$, and we obtain $\sum_{a\in \mathcal{A}_{<21}}n_a=127$ linearly independent polynomials that form $B_1$. 
\end{ex}

We consider now another set of homogeneous polynomials that will be linearly independent from $B_1$ and whose polynomials evaluate to $\fq$. We start with the case $d\leq q^s-1$, which is easier. Let us focus first on the cyclotomic sets $\II_a$ with $a\in\mathcal{A}_{\leq d}\setminus \mathcal{A}_{<d}$. Having $\II_a\cap \Delta_d\neq \emptyset$ implies that the corresponding homogeneous traces $\mathcal{T}_a^h(\xi_{a}^r x_1^{a_1}x_2^{a_2})$, $0\leq r\leq n_a-1$, with $\xi_a$ a primitive element in $\F_{q^{n_a}}$, have at least one monomial which is not divisible by $x_0$. Hence, although the evaluation of these traces in $[\{1\}\times\F_{q^s}^2]$ is going to be equal to the evaluation of $\mathcal{T}_a(\xi_{a}^r x_1^{a_1}x_2^{a_2})$ in $\F_{q^s}^2$, which has coordinates in $\fq$, the evaluation in $[\{0\}\times \{1\}\times\F_{q^s}]$ and $[0:0:1]$ does not necessarily have its coordinates in $\fq$, and, by Lemma \ref{lemassc}, these polynomials might not evaluate to $\fq$. By Lemma \ref{lemassc} and Theorem \ref{baseafin} in one variable, if a polynomial $f(x_0,x_1,x_2)$ evaluates to $\fq$ in $P^2$, $f(0,1,x_2)$ must be a linear combination of traces in the variable $x_2$. A natural idea is to consider linear combinations of homogenized traces such that, when setting $x_0=0,x_1=1$, we obtain that the evaluation of $f(0,1,x_2)$ in $\fqs$ is the same as some trace in the variable $x_2$. To do that, we introduce the following definition.

\begin{defn}
For each $a_2\in \mathcal{A}^1_{\leq d}$, we define the set
$$
Y_{a_2}:=\{a\in \mathcal{A}_{\leq d}\mid \II_a=\II_{(\overline{d-c_2},c_2)}\text{ for some } c_2\in \II_{a_2} \}.
$$
\end{defn}

\begin{rem}\label{ordenciclo}
Recall that, with the order chosen for the cyclotomic sets, we have that $c\in \mathcal{A}$ implies $c_2\in \mathcal{A}^1$. Therefore, in this case $c\in Y_{a_2}$ implies $c_2=a_2$.  
\end{rem}

\begin{ex}\label{ejtoy}
Let us continue with the setting of Example \ref{ejemploprimario1}. We have $d=21$ and $\overline{d}=6$, and we will compute $Y_{a_2}$ for $a_2=0,1$. To do so, we consider first the univariate cyclotomic sets:
$$
\II_0=\{0\},\II_1=\{1,2,4,8\},\II_3=\{3,6,9,12\},\II_5=\{5,10\},\II_7=\{7,11,13,14\},\II_{15}=\{15\}.
$$
For $a_2=0$, we just have $Y_0=\{(3,0)\}$ because $\II_{(3,0)}=\II_{(6,0)}=\II_{(\overline{d-0},0)}$. For $a_2=1$, we need to obtain the minimal elements of the cyclotomic sets $\II_{(\overline{21-1},1)},\II_{(\overline{21-2},2)},\II_{(\overline{21-4},4)}$ and $\II_{(21-8,8)}$. We have
$$
\begin{aligned}
\II_{(5,1)}&=\{ (5,1),(10,2),(5,4),(10,8) \},\\
\II_{(4,2)}&=\{ (2,1),(4,2),(8,4),(1,8) \},\\
\II_{(2,4)}&=\{ (8,1),(1,2),(2,4),(4,8) \}, \\
\II_{(13,8)}&=\{ (11,1),(7,2),(14,4),(13,8) \}.
\end{aligned}
$$
Hence, $Y_1=\{(2,1),(5,1),(8,1),(11,1)\}$. 
\end{ex}

The idea behind the definition of $Y_{a_2}$ is the following: if we consider $c\in Y_{a_2}$ and the polynomial $\mathcal{T}_c^h(\xi_c^r x_1^{c_1}x_2^{c_2})$, then, if $\overline{d-c_2}=d-c_2$, we have the monomial $x_1^{d-c_2}x_2^{c_2}$ in the support of this homogenized trace (if $\overline{d-c_2}<d-c_2$, we would have the monomial $x_0^{q^s-1}x_1^{\overline{d-c_2}}x_2^{c_2}$ instead), and when setting $x_0=0$ and $x_1=1$, we obtain the monomial $x_2^{c_2}$, with $c_2\in \II_{a_2}$, in the support of $f(0,1,x_2)$. We have
\begin{equation}\label{condicionunion}
\overline{d-c_2}=d-c_2 \iff d-c_2\leq q^s-1 \iff d-(q^s-1)\leq c_2.
\end{equation}

In fact, it is clear that all the monomials that we obtain from this homogenized trace when setting $x_0=0,x_1=1$, are monomials $x_2^{c_2}$ with $c_2\in \II_{a_2}$. Thus, the traces associated to $c\in Y_{a_2}$ give monomials $x_2^{c_2}$ with $c_2\in \II_{a_2}$ when setting $x_0=0,x_1=1$.

The case with $\II_{a_2}=\II_{\overline{d}}$ is slightly more complicated, since in this case we have two monomials, $x_1^{q^s-1}x_2^{\overline{d}}$ and $x_2^{d}$ (if $d\geq q^s$), of degree $d$ with different evaluation in $P^2$ which give the same monomial $x_2^{\overline{d}}$ when setting $x_0=0,x_1=1$. This means that two different homogenized traces from different cyclotomic sets can have $x_2^{\overline{d}}$ in its support. We will exclude this case in what follows now as we will study this case separately later. Hence, for a given $a_2\in \mathcal{A}^1_{\leq d}$ with $\II_{a_2}\neq \II_{\overline{d}}$ and $\xi_{a_2}$ a primitive element in $\F_{q^{n_{a_2}}}$, we can consider the sum
$$
\funf=\sum_{c\in Y_{a_2}}\mathcal{T}_c^h(\xi_{a_2}^r x_1^{c_1}x_2^{c_2}),
$$
for $0\leq r\leq n_{a_2}$, and, due to the previous discussion, we obtain that in the support of $\funf(0,1,x_2)$ there are only monomials of the form $x_2^{\gamma_2}$ with $\gamma_2\in \II_{a_2}$. Each monomial $x_2^{\gamma_2}$ can only come from one of the homogenized traces since, if $\gamma_2\neq \overline{d}$, this monomial can only come from the monomial $x_1^{d-\gamma_2}x_2^{\gamma_2}$ in the support of $\funf$, with $\gamma_2\geq d-(q^s-1)$ due to (\ref{condicionunion}). Moreover, the coefficient of each of these monomials $x_2^{\gamma_2}$ is the same that this monomial would have in $\mathcal{T}_{a_2}(\xi_{a_2}^rx_2^{a_2})$ because we saw in Remark \ref{ordenciclo} that $c_2=a_2$ for every $c\in Y_{a_2}$. If $d\leq q^s-1$, the condition from Equation (\ref{condicionunion}) is always satisfied for any $\gamma_2\in \II_{a_2}$. In this case, if we have
$$
\bigcup_{c_2\in \II_{a_2}} \II_{(d-c_2,c_2)}\subset \Delta_{\leq d},
$$
then $Y_{a_2}$ contains all the minimal elements $\gamma\in \mathcal{A}$ such that $\II_\gamma=\II_{(d-\gamma_2,\gamma_2)}$. Therefore, we have all the monomials $x_1^{d-\gamma_2}x_2^{\gamma_2}$, for $\gamma_2\in \II_{a_2}$, in the support of $\funf$, and we obtain $\funf(0,1,x_2)=\mathcal{T}_{a_2}(\xi_{a_2}^rx_2^{a_2})$. The polynomials $\funf$ are homogeneous of degree $d$ and, by Lemma \ref{lemassc}, they evaluate to $\fq$. Thus, their evaluation is in $\PRM_d^\sigma(2)$.

For $d\geq q^s$, we can consider instead the condition
\begin{equation}\label{condy}
\bigcup_{c_2\in \II_{a_2},c_2>d-(q^s-1)} \II_{(d-c_2,c_2)}\subset \Delta_{\leq d}.
\end{equation}
We avoid the case $c_2=d-(q^s-1)=\overline{d}$ as we will study it later, and we consider only $c_2> d-(q^s-1)$ in order to satisfy Equation (\ref{condicionunion}). Reasoning as in the previous case, if the previous condition is satisfied, then $\funf(0,1,x_2)$ has in its support all the terms from $\mathcal{T}_{a_2}(\xi_{a_2}^rx_2^{a_2})$ with degree greater than $d-(q^s-1)=\overline{d}$. We claim that, in this situation, it is always possible to construct a polynomial $\fung$ whose evaluation is in $\PRM_d^\sigma(2)$ and such that $\fung(1,x_1,x_2)=\funf(1,x_1,x_2)$, $\fung(0,1,x_2)=\mathcal{T}_{a_2}(\xi_{a_2}^rx_2^{a_2})$, and $\fung(0,0,1)=0$.

We first note that, in this situation, we can homogenize the equations of the field and obtain homogeneous polynomials of degree $d$. By this, what we mean is that we can consider a multiple of $x_i^{q^s}-x_i$, for $i=1,2$, and homogenize it up to degree $d$. If this multiple has degree less than $d$, then that homogenized polynomial evaluates to the 0 vector in $P^2$. However, when the degree of this multiple is exactly equal to $d\geq q^s$, we can obtain the following polynomials by multiplying the field equations by monomials and then homogenizing:
$$
\left(x_1^{c_1}x_2^{c_2-1}(x_2^{q^s}-x_2) \right)^h=\left( x_1^{c_1}x_2^{c_2+q^s-1}-x_1^{c_1}x_2^{c_2} \right)^h=x_1^{c_1}x_2^{c_2+q^s-1}-x_0^{q^s-1}x_1^{c_1}x_2^{c_2},
$$
where we are assuming that $c_1+c_2+q^s-1=d$ and $c_2>0$. We note that we only consider $d\leq 2(q^s-1)$ (for a higher degree $\PRM_d(2)$ is the whole space). Thus, $c_1+c_2=\overline{d}$. Using the other field equation, we can get
$$
\left(x_1^{c_1-1}x_2^{c_2}(x_1^{q^s}-x_1) \right)^h=\left( x_1^{c_1+q^s-1}x_2^{c_2}-x_1^{c_1}x_2^{c_2} \right)^h=x_1^{c_1+q^s-1}x_2^{c_2}-x_0^{q^s-1}x_1^{c_1}x_2^{c_2},
$$

All of these polynomials are equivalent to $x_1^{c_1}x_2^{c_2}(1-x_0)$ in $S/I(P^2)$, which is a more compact way of writing them, and we will refer to them as \textit{homogenized field equations}. Although this last expression is not homogeneous, it has the same evaluation in $P^2$ as a homogeneous polynomial of degree $d$, which implies that its evaluation is also in $\PRM_d(2)$. With this in mind, we have that, for any $0\leq c_2\leq \overline{d}-1$, the polynomial $x_1^{\overline{d}-c_2}x_2^{c_2}(1-x_0)$ can be seen as a homogeneous polynomial of degree $d$, and its evaluation in $[\{1\}\times\fqs^2]$ is the zero vector, in $[\{0\}\times \{1\}\times\fqs]$ it is the same as the evaluation of $x_2^{c_2}$, and it is $0$ in $[0:0:1]$. Moreover, the polynomial $x_1x_2^{c_2}(1-x_0)$ has the same evaluation. For $c_2=\overline{d}$, we have the polynomial $x_2^{\overline{d}}(1-x_0)$, but in this case the evaluation at $[0:0:1]$ of this polynomial is equal to 1. This polynomial will only be considered later when we study the case with $\II_{a_2}=\II_{\overline{d}}$.

As a consequence, if we add to $\funf$ a homogenized field equation, the evaluation of the resulting polynomial in $[\{1\}\times\fqs^2]$ does not change, and when setting $x_0=0,x_1=1$, we obtain $\funf(0,1,x_2)+x_2^{c_2}$, for some $0\leq c_2\leq \overline{d}-1$. Hence, if $\II_{a_2}\neq \II_{\overline{d}}$, and if we have the condition $\bigcup_{c_2\in \II_{a_2},c_2>d-(q^s-1)} \II_{(d-c_2,c_2)}\subset \Delta_{\leq d}$ (we recall that, under this assumption, $\funf(0,1,x_2)$ has in its support all the terms from $\mathcal{T}_{a_2}(\xi_{a_2}^rx_2^{a_2})$ with degree greater than $\overline{d}$), then, adding adequate multiples of the homogenized field equations, we can obtain a polynomial $\fung$ such that $\fung(1,x_1,x_2)=\funf(1,x_1,x_2)$, $\fung(0,1,x_2)=\mathcal{T}_{a_2}(\xi_{a_2}^rx_2^{a_2})$, and $\fung(0,0,1)=0$. Therefore, the polynomial $\fung$ is defined as the polynomial obtained by adding the necessary multiples of the homogenized field equations to $\funf$ to obtain $\fung(0,1,x_2)=\mathcal{T}_{a_2}(\xi_{a_2}^rx_2^{a_2})$. Because of all the previous discussion, it is clear that the evaluation of $\fung$ is in $\PRM_d^\sigma(2)$.

Moreover, we see that the polynomial
$$
\funh=x_0\left(\sum_{c\in Y_{a_2}}\mathcal{T}_c(\xi_{a_2}^r x_1^{c_1}x_2^{c_2})\right)+(1-x_0)x_1\mathcal{T}_{a_2}(\xi_{a_2}^r x_2^{a_2})
$$
has the same evaluation as the polynomial $\fung$, which means that its evaluation is also in $\PRM_d^\sigma(2)$. 

Furthermore, avoiding the case in which $\II_{a_2}=\II_{\overline{d}}$, we can express both the case with $d\geq q^s$ and $d\leq q^s-1$ using the same polynomials and conditions. To see this, we first introduce the following notation:
$$
Y=\left\{a_2\in \mathcal{A}^1_{\leq d}\mid \II_{a_2}\neq \II_{\overline{d}} \text{ such that } \bigcup_{c_2\in \II_{a_2},c_2>d-(q^s-1)}\II_{(d-c_2,c_2)}\subset \Delta_{\leq d}\right\}.
$$
The elements of $Y$ are just the $a_2\in \mathcal{A}^1_{\leq d}$ such that we can construct a polynomial in $\PRM_d^\sigma(2)$ whose evaluation in $[\{0\}\times \{1\}\times\fqs]$ is equal to some trace of $x_2^{a_2}$ with the previous ideas. In the case $d\leq q^s-1$, the condition in the set $Y$ is the same that we were considering before. Note that for $a_2=0$ and $d=q^s-1$, the condition that we had for $d\leq q^s-1$ was 
$$
\bigcup_{c_2\in \II_{a_2}} \II_{(d-c_2,c_2)}=\II_{(q^s-1,0)}=\{(q^s-1,0)\}\subset \Delta_{\leq q^s-1},
$$
which is always satisfied. The condition that we have used for $Y$ when $a_2=0$ and $d=q^s-1$ would be 
$$
\bigcup_{c_2\in \II_{a_2},c_2>d-(q^s-1)}\II_{(d-c_2,c_2)}=\emptyset \subset \Delta_{\leq d},
$$
which is always satisfied as well. The following result summarizes the previous discussion. 

\begin{lem}\label{lemab1b2}
Let $1\leq d\leq 2(q^s-1)$, and let $\xi_{a_2}$ be a primitive element in $\F_{q^{n_{a_2}}}$. The evaluation of the polynomials in the set
$$
\begin{aligned}
B_2=
&\bigcup_{a_2\in Y}\left\{x_0\left(\sum_{c\in Y_{a_2}}\mathcal{T}_c(\xi_{a_2}^r x_1^{c_1}x_2^{c_2})\right)+(1-x_0)x_1\mathcal{T}_{a_2}(\xi_{a_2}^r x_2^{a_2}),0\leq r\leq n_{a_2}-1\right\}
\end{aligned}
$$
is in $\PRM_d^\sigma(2)$. Moreover, the evaluation of the polynomials in $B_1\cup B_2$ is linearly independent.
\end{lem}
\begin{proof}
In the previous discussion we have showed that, if $d\geq q^s$, all the polynomials in $B_2$ have their evaluation in $\PRM_d(2)$, and we also checked that they evaluate to $\fq$ due to Lemma \ref{lemassc}. For the case $d\leq q^s-1$, these polynomials have the same evaluation as $\funf$, which means that their evaluation is also in $\PRM_d^\sigma(2)$.

The evaluation of the polynomials in $B_2$ is linearly independent since it is linearly independent in $[\{0\}\times \{1\}\times\fqs]$ by the affine case from Theorem \ref{baseafin}: in $[\{0\}\times \{1\}\times\fqs]$ we have univariate traces in $x_2$ from different cyclotomic sets. Moreover, the evaluation of the polynomials in $B_2$ is linearly independent from the evaluation of the polynomials in $B_1$ because the evaluation of the polynomials in $B_1$ is zero in $[\{0\}\times \{1\}\times\fqs]$.
\end{proof}

\begin{rem}\label{remfacil}
Let $a_2\in \mathcal{A}^1$, and let 
$$
Y'_{a_2}:=\{a\in \mathcal{A}_{\leq d}\setminus \mathcal{A}_{<d}\mid \II_a=\II_{(\overline{d-c_2},c_2)}\text{ for some } c_2\in \II_{a_2} \}.
$$
The set
$$
\begin{aligned}
B'_2=
&\bigcup_{a_2\in Y}\left\{x_0\left(\sum_{c\in Y'_{a_2}}\mathcal{T}_c(\xi_{a_2}^r x_1^{c_1}x_2^{c_2})\right)+(1-x_0)x_1\mathcal{T}_{a_2}(\xi_{a_2}^r x_2^{a_2}),0\leq r\leq n_{a_2}-1\right\}
\end{aligned}
$$
has the same properties as $B_2$ in Lemma \ref{lemab1b2}. This is because, for any $a\in\mathcal{A}_{<d}$, we have already considered $x_0\mathcal{T}_a(\xi_a^rx_1^{a_1}x_2^{a_2})$, $0\leq r \leq n_a-1$, in $B_1$, and $x_0 \mathcal{T}_a(\xi_{a_2}^rx_1^{a_1}x_2^{a_2})$ is in the span of those traces for any $0\leq r\leq n_{a_2}-1$.\end{rem}

\begin{ex}\label{ejemploprimario2}
Let us continue with the setting from Example \ref{ejtoy} and compute the polynomials in the set $B_2'$ defined in Remark \ref{remfacil}, although we will also compute all the sets needed to obtain $B_2$ as well. We first compute $Y$. We have that $a_2\in Y$ if the condition (\ref{condy}) is verified. In this case, $d=21$ and $d-(q^s-1)=\overline{d}=6$. For $a_2=0$ we have $\II_0=\{0\}$, and the union of cyclotomic sets in the left hand side of (\ref{condy}) is empty, which means that the condition is satisfied, and $0\in Y$.

For $a_2=1$, we verify that $\{(11,1),(7,2),(13,8),(14,4)\}=\II_{(21-8,8)}\subset \Delta_{\leq 21}$ (note that $8$ is the only element in $\II_1$ greater than $\overline{d}=6$). The condition (\ref{condy}) is satisfied and $1\in Y$. We do not consider $a_2=3$ now since $\II_3=\II_{\overline{d}}$, which is the case that we will cover in Example \ref{ejemploprimario3}. For $a_2\in \{5,7,15\}$, it is easy to check that we have $a_2\not\in Y$. For example, for $a_2=7$, the cyclotomic set $\II_{(21-7,7)}=\{(14,7),(7,11),(11,13),(13,14)\}\not\subset \Delta_{\leq 21}$, because, for instance, $(11,13)\not \in \Delta_{\leq 21}$. Therefore, we have 
$$
Y=\{0,1\}.
$$
Now, for each $a_2\in Y$, we have to compute $Y_{a_2}$. This was already done in Example \ref{ejtoy}, and $Y_0=\{(3,0)\}$ and $Y_1=\{(2,1),(5,1),(8,1),(11,1)\}$. By Remark \ref{remfacil}, we can consider the sets $Y'_0=\emptyset$ and $Y'_1=\{(11,1)\}$ ($\II_{(11,1)}$ is the only cyclotomic set that we have considered which is in $\Delta_{\leq 21}\setminus \Delta_{<21}$) instead of $Y_0,Y_1$, respectively, and the set $B'_2$ obtained satisfies the same properties as $B_2$. For simplicity, we construct $B'_2$ instead of $B_2$.

We now obtain the polynomials in $B'_2$. For $a_2=0$ we have $n_{a_2}=n_0=1$, which means that we only consider one polynomial, and we also have $Y'_0=\emptyset$. We consider the following polynomial:
$$
\{ (1-x_0)x_1\mathcal{T}_0(x_2^0)\}=\{(1-x_0)x_1\}.
$$

For the case $a_2=1$, we have $n_{a_2}=n_1=4$, and we have $Y'_1=\{(11,1)\}$. Thus, using Remark \ref{remfacil}, we consider the set of polynomials
$$
\{x_0\mathcal{T}_{(11,1)}(\xi_{1}^rx_1^{11}x_2)+(1-x_0)x_1\mathcal{T}_1(\xi_1^r x_2),0\leq r\leq n_1-1 \},
$$
where $\xi_{1}$ is a primitive element in $\F_{q^{n_1}}=\F_{16}$. Hence, we have constructed the set
$$
B'_2=\{(1-x_0)x_1\}\cup\{x_0\mathcal{T}_{(11,1)}(\xi_{1}^rx_1^{11}x_2)+(1-x_0)x_1\mathcal{T}_1(\xi_1^r x_2),0\leq r\leq n_1-1 \},
$$
whose size is $n_1+n_0=5$. In Example \ref{ejemploprimario1} we obtained that the cardinality of $B_1$ is 127. This means that $B_1\cup B'_2$ (and $B_1\cup B_2$) contains 132 polynomials whose evaluation is in $\PRM_{21}^\sigma(2)$, and the evaluation of these polynomials is linearly independent.
\end{ex}

We construct now one last set $B_3$. In the previous study, we have omitted the case in which $\II_{a_2}=\II_{\overline{d}}$. Therefore, we consider now $a_2\in\mathcal{A}^1$ be such that $\II_{a_2}=\II_{\overline{d}}$. We assume that $a_2\in \mathcal{A}^1_{\leq d}$ (if $a_2\not \in\mathcal{A}^1_{\leq d}$ the set $B_3$ will be the empty set). We follow a very similar reasoning to the one we did for the set $B_2$. For the case $1\leq d\leq q^s-1$, we were considering the polynomials 
$$
\funf=\sum_{c\in Y_{a_2}}\mathcal{T}_c^h(\xi_{a_2}^r x_1^{c_1}x_2^{c_2})
$$
to construct $B_2$. We can still consider such a polynomial if $\II_{a_2}=\II_d$, but in this case, $\funf(0,0,1)$ is the coefficient of $x_2^d$ in $\funf$, which is nonzero if $\II_{(0,d)}\subset \Delta_{\leq d}$. We have that $\funf(0,0,1)\in \fq$ only if $r=0$, and in that case the polynomial 
$$
\funl=x_0\left(\sum_{c\in Y_{a_2}}\mathcal{T}_c( x_1^{c_1}x_2^{c_2})\right)+(1-x_0)x_1\mathcal{T}_{a_2}(x_2^{a_2})+(1-x_0)(1-x_1)x_2^d
$$
has the same evaluation in $P^2$ as $f_{a_2}^0$. If $ \bigcup_{c_2\in \II_{a_2}}\II_{(d-c_2,c_2)}\subset \Delta_{\leq d}$, i.e., we have $f_{a_2}^0(0,1,x_2)=\mathcal{T}_{a_2}(x_2^{a_2})=\mathcal{T}_d(x_2^d)$, $\funl$ evaluates to $\fq$ and its evaluation is in $\PRM_d(2)$ (it has the same evaluation as $\funf$).

For the case $d\geq q^s$, we can consider the homogenized field equation $x_2^{\overline{d}}(1-x_0)$ to obtain a polynomial $\fung$ such that $g^r_{a_2}(1,x_1,x_2)=\funf(1,x_1,x_2)$ and $\fung(0,1,x_2)=\mathcal{T}_{a_2}(\xi_{a_2}^rx_2^{a_2})$. The problem that arises in this specific case is the following: the monomial $x_2^{\overline{d}}$ can be obtained when setting $x_0=0,x_1=1,$ from the monomials $x_1^{q^s-1}x_2^{\overline{d}}$ and $x_2^d$, both of them of degree $d$. Hence, following the previous notation, we have to study two different cases: if $\funf(0,1,x_2)$ has $x_2^{\overline{d}}$ in its support (which means that $x_1^{q^s-1}x_2^{\overline{d}}$ is in the support of $f$), or if $\funf(0,1,x_2)$ does not have $x_2^{\overline{d}}$ in its support. 

We start with the case in which $\funf(0,1,x_2)$ does not have $x_2^{\overline{d}}$ in its support, where we need to use $x_2^{\overline{d}}(1-x_0)$ to construct $\fung$. The main difference is that in this case $\fung(0,0,1)$ is equal to the coefficient of $x_2^{\overline{d}}$, which is nonzero. Therefore, by Lemma \ref{lemassc}, this coefficient has to be in $\fq$ if $\fung$ evaluates to $\fq$. We are also interested in obtaining $\fung(0,1,x_2)=\mathcal{T}_{a_2}(\xi_{a_2}^r x_2^{a_2})$ for some $0\leq r\leq n_{a_2}-1$. The coefficient of $x_2^{\overline{d}}$ in $\fung(0,1,x_2)$ is precisely the coefficient with which we considered $x_2^{\overline{d}}(1-x_0)$ when constructing $\fung$. Thus, the only possibility to have this coefficient in $\fq$ is that this coefficient is equal to 1 (the case $r=0$), and $g_{a_2}^0(0,1,x_2)=\mathcal{T}_{a_2}(x_2^{a_2})$. With this in mind, it is easy to check that $\funl$, as defined previously, has the same evaluation as the polynomial $g_{a_2}^0$ in $P^2$ in this case. As we argued for the set $B_2$, to ensure that the evaluation of $\funl$ is in $\PRM_d(2)$, we need to have $ \bigcup_{c_2\in \II_{a_2},c_2>d-(q^s-1)}\II_{(d-c_2,c_2)}\subset \Delta_{\leq d}$. This condition ensures that $f_{a_2}^0(0,1,x_2)$ has all the monomials from $\mathcal{T}_{a_2}(x_2^{a_2})$ in its support, except maybe the monomials $x_2^{c_2}$ with $c_2\in \{0,1,\dots,\overline{d}\}$, which appear in the support of $g_{a_2}^0(0,1,x_2)$ when adding to $f_{a_2}^0(0,1,x_2)$ the corresponding homogenized field equations.

Finally, we consider the case in which we have $x_1^{q^s-1}x_2^{\overline{d}}$ in the support of $\funf$, i.e., $\funf(0,1,x_2)$ has $x_2^{\overline{d}}$ in its support. If we look at the definition of $\funf$, this happens if and only if $\II_{(q^s-1,\overline{d})}\subset \Delta_{\leq d}$. This is equivalent to having that $\overline{d}$ is the maximal element of $\II_{a_2}$. Therefore, the condition $ \bigcup_{c_2\in \II_{a_2},c_2>d-(q^s-1)}\II_{(d-c_2,c_2)}=\emptyset\subset \Delta_{\leq d}$ is automatically satisfied in this case. This allows us to construct a polynomial $\funll$ which is very similar to $\funl$:
$$
\funll=\funl -x_0\mathcal{T}_{(q^s-1,a_2)}(x_1^{q^s-1}x_2^{a_2}).
$$
Indeed, we can subtract the polynomial $\mathcal{T}^h_{(q^s-1,\overline{d})}(x_1^{c_1}x_2^{c_2})$ from $f_{a_2}^0$, and, adding the corresponding homogenized field equations (we will need to use $x_2^{\overline{d}}(1-x_0)$ in order to obtain $\mathcal{T}_{a_2}(x_2^{a_2})$ when setting $x_0=0,x_1=1$, as we have subtracted the monomial $x_1^{q^s-1}x_2^{\overline{d}}$), we would get a polynomial $\fungg$ such that $\fungg(1,x_1,x_2)=f_{a_2}^0(1,x_1,x_2)-\mathcal{T}_{(q^s-1,a_2)}(x_1^{q^s-1}x_2^{a_2})$, $\fungg(0,1,x_2)=\mathcal{T}_{a_2}(x_2^{a_2})$, $\fungg(0,0,1)=1$. Hence, the polynomial $\funll$ has the same evaluation as the polynomial $\fungg$, which means that the evaluation of $\funll$ is in $\PRM_d^\sigma(2)$. 

On the other hand, we saw previously that the condition $ \bigcup_{c_2\in \II_{a_2},c_2>d-(q^s-1)}\II_{(d-c_2,c_2)}\subset \Delta_{\leq d}$ is satisfied in this case. Hence, adding homogenized field equations to $\funf$ as we did to obtain the set $B_2$, we can obtain a polynomial $\fung$ such that $\fung(1,x_1,x_2)=\funf(1,x_1,x_2)$, $\fung(0,1,x_1)=\mathcal{T}_{a_2}(\xi_{a_2}^r x_2^{a_2}), \fung(0,0,1)=0$. Note that in this case we are not using the homogenized field equation $x_2^{\overline{d}}(1-x_0)$ to construct $\fung$ since we already have the monomial $x_1^{q^s-1}x_2^{\overline{d}}$ in the support of $\funf$, which reduces to $x_2^{\overline{d}}$ when setting $x_0=0,x_1=1$. It is easy to check that the polynomial 
$$
\funh=x_0\left(\sum_{c\in Y_{a_2}}\mathcal{T}_c(\xi_{a_2}^r x_1^{c_1}x_2^{c_2})\right)+(1-x_0)x_1\mathcal{T}_{a_2}(\xi_{a_2}^rx_2^{a_2}),
$$
where $\xi_{a_2}$ is a primitive element in $\F_{q^{n_{a_2}}}$, has the same evaluation in $P^2$ as $\fung$. Therefore, the evaluation of the polynomials $\funh$ is equivalent modulo $S/I(P^2)$ to the evaluation of some homogeneous polynomials of degree $d$, and they evaluate to $\fq$, which means that the evaluation of the polynomials $\funh$ is in $\PRM_d^\sigma(2)$. We can now define the set $B_3$ in the following way:

\begin{enumerate}
    \item[(a)] If $\II_{(q^s-1,\overline{d})}\subset \Delta_{\leq d}$, we set $B_3=\{\funl - x_0\mathcal{T}_{(q^s-1,a_2)}(x_1^{q^s-1}x_2^{a_2})\}\cup \{h_{a_2}^r, 0\leq r\leq n_{a_2} -1\} $.
    \item[(b)] If $\II_{(q^s-1,\overline{d})}\not\subset \Delta_{\leq d}$:
    \begin{enumerate}
        \item[(b.1)] If $\condicion$, we set $B_3=\{\funl\}$.
        \item[(b.2)] We set $B_3=\emptyset$ otherwise.
    \end{enumerate}
\end{enumerate}

With this definition, we can summarize everything discussed thus far in the following result.

\begin{lem}\label{lemab1b2b3}
Let $1\leq d\leq 2(q^s-1)$ and let $a_2\in \mathcal{A}^1$ such that $\II_{a_2}=\II_{\overline{d}}$. If $B_3\neq \emptyset$, the evaluation of the set $B_3$ is in $\PRM_d^\sigma(2)$, and the evaluation of the set $B=B_1\cup B_2\cup B_3$ is linearly independent.
\end{lem}
\begin{proof}
In the previous discussion we have seen that, under the stated conditions, the evaluation of the polynomials in $B_3$ is in $\PRM_d^\sigma(2)$, i.e., for each polynomial in $B_3$, a homogeneous polynomial of degree $d$ with the same evaluation can be constructed, and it evaluates to $\fq$. 

The set $B_1\cup B_2$ is linearly independent due to Lemma \ref{lemab1b2}. The polynomial $\funl$ (respectively, the polynomial $\funl - x_0\mathcal{T}_{(q^s-1,a_2)}(x_1^{q^s-1}x_2^{a_2})$) is not contained in the span of $B_1\cup B_2$ since it is the only polynomial that we are considering with nonzero evaluation at $[0:0:1]$. With this in mind, the same argument as in Lemma \ref{lemab1b2} proves that the evaluation of the rest of polynomials in $B_3$ (if any) is linearly independent, and the evaluation of these polynomials is also linearly independent with the evaluation of the polynomials in $B_1\cup B_2$.
\end{proof}

\begin{rem}\label{remfacil2}
We can argue as in Remark \ref{remfacil} to construct simpler polynomials than the polynomials $h_{a_2}^r$ and $l_{a_2}$. This gives rise to a set $B_3'$ with the properties stated in Lemma \ref{lemab1b2b3}.
\end{rem}

\begin{ex}\label{ejemploprimario3}
Let us continue with the setting from \ref{ejemploprimario2}. We did not study the case $a_2=3$ because $\II_{a_2}=\II_3=\II_{\overline{d}}=\II_6$. This case is covered by Lemma \ref{lemab1b2b3}, and we construct the set $B'_3$ from Remark \ref{remfacil2} in this example. Following the statement of Lemma \ref{lemab1b2b3}, we check first if $\II_{(q^s-1,\overline{d})}\subset\Delta_{\leq d}$, for $d=21$, $\overline{d}=6$ and $q^s-1=15$. We have
$$
\II_{(15,6)}=\{(15,3),(15,6),(15,9),(15,12) \}.
$$
We see that $\II_{(15,6)}\not\subset\Delta_{\leq 21}$, for example we have $(15,9)$ with $15+9=24>21$.

Now we have to verify the condition (\ref{condy}). The only elements $c_2$ in $\II_{a_2}=\{ 3,6,9,12\}$ such that $c_2>\overline{d}$ are 9 and 12. The corresponding cyclotomic sets $\II_{(21-9,9)}$ and $\II_{(21-12,12)}$ are
$$
\begin{aligned}
\II_{(9,3)}=\{(9,3),(3,6),(12,9),(6,12)\},\\
\II_{(6,3)}=\{(6,3),(12,6),(3,9),(9,12) \}.
\end{aligned}
$$

Hence, we see that the condition (\ref{condy}) is satisfied since both cyclotomic sets are contained in $\Delta_{\leq 21}$. Therefore, we have to construct $\funl$, for which we have to compute $Y_3$. We have $\II_{(21-6,6)}=\II_{(15,3)}$ from before, but we have seen that this cyclotomic set is not contained in $\Delta_{\leq 21}$. Thus, $(15,3)\not \in Y_3$. On the other hand, we have just seen that $(6,3),(9,3)\in Y_{a_2}$, as both of them are contained in $\Delta_{\leq 21}$. The last cyclotomic set that we have to consider is the following:
$$
\begin{aligned}
\II_{(\overline{21-3},3)}=\{(3,3),(6,6),(9,9),(12,12) \},
\end{aligned}
$$
which is not contained in $\Delta_{\leq 21}$. Hence, $Y_3=\{(6,3),(9,3) \}$. Using Remarks \ref{remfacil} and \ref{remfacil2} in this case gives $Y'_3=Y_3$, which means that we have $B'_3=B_3$. The only polynomial in $B_3$ is
$$
l_{3}=x_0\left( \mathcal{T}_{(9,3)}(x_1^9x_2^3)+\mathcal{T}_{(6,3)}(x_1^6x_2^3)\right)+(1-x_0)x_1\mathcal{T}_3(x_2^3)+(1-x_0)(1-x_1)x_2^{21}.
$$

We obtain that there are 133 polynomials in $B_1\cup B_2 \cup B_3$, with linearly independent evaluation, and this evaluation is in $\PRM_{21}^\sigma(2)$.
\end{ex}

The following results show that the case where $1\leq d\leq q^s-1$ is particularly simple. 

\begin{lem}\label{lemaciclotomicostamano1}
Let $1\leq d\leq q^s-1$ .We have that $\abs{I_d}=1$ if and only if $d=\lambda\frac{q^s-1}{q-1}$, for some integer $1\leq \lambda\leq q-1$.
\end{lem}
\begin{proof}
We only need to observe that
$$
\begin{aligned}
\abs{I_d}=1 &\iff d q\equiv d\bmod q^s-1\iff d(q-1)=\lambda(q^s-1)=\lambda(q-1)\frac{q^s-1}{q-1} \\
&\iff d=\lambda\frac{q^s-1}{q-1}, \text{ for some } 1\leq \lambda\leq q-1.
\end{aligned}
$$
\end{proof}

\begin{prop}
Let $1\leq d\leq q^s-1$. Then $B_3\neq \emptyset$ if and only if $d$ is a multiple of $\frac{q^s-1}{q-1}$. In that situation 
$$
B_3=\{x_2^d\}.
$$
\end{prop}
\begin{proof}
If $d$ is a multiple of $\frac{q^s-1}{q-1}$, by Lemma \ref{lemaciclotomicostamano1}, we have that $\abs{\II_d}=1$ and $\II_{(0,d)}\subset \Delta_{\leq d}$. By Lemma \ref{lemab1b2b3}, $B_3=\{l_{d}\}$. We have $Y_{d}=\{(0,d)\}$ from its definition. Then, by the definition of $l_{d}$ we have $l_{d}=x_0\mathcal{T}_{(0,d)}(x_2^d)+(1-x_0)x_1\mathcal{T}_d(x_2^d)+(1-x_0)(1-x_1)x_2^d=x_0x_2^d+(1-x_0)x_1x_2^d+(1-x_0)(1-x_1)x_2^d=x_2^d$. 

On the other hand, if $B_3\neq\emptyset$ and we consider $a_2\in \mathcal{A}^1_{\leq d}$ with $\II_{a_2}=\II_d$, by Lemma \ref{lemab1b2b3} we have that $\condicionn$. Using Lemma \ref{lemaciclotomicostamano1}, we assume that $\abs{\II_{a_2}}>1$, and we will obtain a contradiction. Let $e\in\II_{a_2}$ with $e\neq d$. This implies that there is an integer $l>0$ such that $d\equiv q^l e\bmod q^s-1$. Therefore, we have $(\overline{q^l(d-e)},d)\in \II_{(d-e,e)}$, with $\overline{q^l(d-e)}\neq 0$. This implies that $\II_{(d-e,e)}\not\subset \Delta_{\leq d}$, a contradiction.
\end{proof}

In order to assert that $B$ is a basis, we would need to show that $B$ generates the whole code $\PRM_d^\sigma(2)$. However, we have already computed the dimension for  $\PRM_d^{\sigma,\perp}(2)$. By Lemma \ref{lemab1b2b3}, we know that the evaluation of the polynomials in $B$ is linearly independent, which means that if we show that $\abs{B}=n-\dim \PRM_d^{\sigma,\perp}(2) $, then this implies that $B$ is a basis. To see this, we will introduce a new decomposition of the sets $B$ and $D$. 

Let $1\leq d\leq 2(q^s-1)$, and $d^\perp=2(q^s-1)-d$. For the set $B$, we first define $\Gamma_1=B_1$. On the other hand, let $a_2\in\mathcal{A}^1$ such that $\II_{a_2}=\II_{\overline{d}}$, and we define $\Gamma_2$ in the following way:
\begin{enumerate}
    \item If $\II_{(q^s-1,\overline{d})}\subset \Delta_{\leq d}$, we set
    $$
    \Gamma_2=B_2\cup \{ h_{a_2}^r, 0\leq r\leq n_{a_2}-1\}.
    $$
    \item We set 
    $$
    \Gamma_2=B_2,
    $$
    otherwise.
\end{enumerate}

And we define $\Gamma_3=B\setminus (\Gamma_1\cup \Gamma_2)$. Equivalently, we consider the following definition:

\begin{enumerate}
    \item[(a)] If $\II_{(q^s-1,\overline{d})}\subset \Delta_{\leq d}$, we set
    $$
    \Gamma_3=\{\funl - x_0\mathcal{T}_{(q^s-1,a_2)}(x_1^{q^s-1}x_2^{a_2})\}.
    $$
    \item[(b)] If $\II_{(q^s-1,\overline{d})}\not\subset \Delta_{\leq d}$:
    \begin{enumerate}
        \item[(b.1)] If $\condicion$, we set
        $$
        \Gamma_3=\{\funl\}.
        $$
        \item[(b.2)] We set 
        $$
        \Gamma_3=\emptyset,
        $$
        otherwise.
    \end{enumerate}
\end{enumerate}

It is clear by construction that $B=\Gamma_1\cup \Gamma_2\cup \Gamma_3$. The idea behind this decomposition is that in $\Gamma_1$ we have sets of size $n_a$ for some $a\in \mathcal{A}$, in $\Gamma_2$ we have sets of size $n_{a_2}$ for some $a_2\in\mathcal{A}^1$, and in $\Gamma_3$ we have a set of size 1 (if any). Now we define a similar decomposition for $D$, and we will see later why we are interested in this decomposition.

For the set $D$, we define first $\Gamma_1^\perp=D_1$. Let $a_2\in \mathcal{A}^1$ such that $\II_{a_2}=\II_{\overline{d}}$. Now we define $\Gamma_3^\perp$ as follows:
\begin{enumerate}
    \item If there is an element $c\in\mathcal{A}$ such that $c_2=a_2$, $\II_c\neq \II_{(0,\overline{d^\perp})}$, and $M_c(d^\perp)$ contains monomials of the two types, we set
    $$
    \Gamma_3^\perp=(x_0-1)(x_1-1).
    $$
    \item We set 
    $$
    \Gamma_3^\perp=\emptyset,  
    $$
    otherwise.
\end{enumerate}

We can now define $\Gamma_2^\perp=D\setminus (\Gamma_1^\perp\cup \Gamma_3^\perp)$. This can also be expressed in the following way:
\begin{equation}\label{gamma2perp}
\Gamma_2^\perp=(D_2\cup D_3\cup D_4)\setminus \{(x_0-1)(x_1-1)\}.
\end{equation}

Again, by construction we have $D=\Gamma_1^\perp\cup \Gamma_2^\perp\cup \Gamma_3^\perp$. 

\begin{rem}\label{condicionnested}
The condition in (1) from the definition of $\Gamma_3^\perp$ implies that $M_{(0,\overline{d^\perp})}(d^\perp)$ contains monomials of the two types. Indeed, if $d^\perp\geq q^s$, $M_{(0,\overline{d^\perp})}(d^\perp)$ always contains monomials of the two types, and if $d^\perp\leq q^s-1$ and there is an element $c\in\mathcal{A}$ such that $c_2=a_2$, $\II_c\neq \II_{(0,\overline{d^\perp})}$, and $M_c(d^\perp)$ contains monomials of the two types, this means that there is $\gamma\in \II_c$ with $\gamma_1>0$ such that $\gamma_1+\gamma_2=d^\perp$ by Lemma \ref{condicionesM}, with $\gamma_2\in \II_{d^\perp}$. Therefore, $d^\perp$ is not the minimal element in $\II_{d^\perp}$, which means that $M_{(0,d^\perp)}(d^\perp)$ contains monomials of the two types. Hence, we have $(x_0-1)(x_1-1)\in \Gamma_3^\perp$ if and only if $(x_0-1)(x_1-1)\in D_3$.
\end{rem}

Let $b_2\in \mathcal{A}^1$ such that $\II_{b_2}=\II_{\overline{d}}$, for some degree $1\leq d\leq 2(q^s-1)$. For ease of use, we recall here the sizes of the set we have just defined:
\begin{enumerate}
    \item[(a.1)] $\abs{\Gamma_1}=\abs{B_1}=\sum_{a\in \mathcal{A}_{<d}}n_a$.
    \item[(a.2)] $\abs{\Gamma_2}=\abs{B_2}+n_{\overline{d}}=\sum_{a_2\in Y}n_{a_2}+n_{\overline{d}}$ if $\II_{(q^s-1,\overline{d})}\subset \Delta_{\leq d }$, and $\abs{\Gamma_2}=\abs{B_2}$ otherwise.
    \item[(a.3)] $\abs{\Gamma_3}=1$ if $\condicionnn$, and $\abs{\Gamma_3}=0$ otherwise.
    \item[(b.1)] $\abs{\Gamma_1^\perp}=\abs{D_1}=\sum_{a\in U}n_a$.
    \item[(b.2)] $\abs{\Gamma_2^\perp}=\abs{D_2}+\abs{D_3\setminus \{(x_0-1)(x_1-1)\}}+\abs{D_4}=\sum_{a_2\in V}n_{a_2}+n_{\overline{d}}+\abs{D_4}$ if $M_{(0,\overline{d^\perp})}(d^\perp)$ contains monomials of the two types, and $\abs{\Gamma_2^\perp}=\sum_{a_2\in V}n_{a_2}+\abs{D_4}$ otherwise, where $\abs{D_4}=1$ if $d=q^s-1$, and $\abs{D_4}=0$ otherwise.
    \item[(b.3)] $\abs{\Gamma_3^\perp} =1$ if there is an element $c\in\mathcal{A}$ such that $c_2=b_2$, $\II_c\neq \II_{(0,\overline{d^\perp})}$, and $M_c(d^\perp)$ contains monomials of the two types, and $\abs{\Gamma_3^\perp} =0$ otherwise.
\end{enumerate}

\begin{defn}
Let $b=(b_1,b_2)\in \zmsdos$. We define
$$
b'=(b'_1,b'_2):=(q^s-1-b_1,q^s-1-b_2).
$$
\end{defn}

\begin{rem}
Let $c\in\mathcal{A}$. Then $c_2\in \II_{a_2}$ if and only if $c'_2=q^s-1-c_2\in \II_{a'_2}$.
\end{rem}

We are interested in doing these decompositions because the length of these codes is $n=\frac{q^{3s}-1}{q^s-1}=q^{2s}+q^s+1$, and we also have $\sum_{a\in \mathcal{A}}n_a=q^{2s}$, $\sum_{a_2\in\mathcal{A}^1}n_{a_2}=q^s$. We prove now that $\abs{\Gamma_1}+\abs{\Gamma_1^\perp}=q^{2s}$, $\abs{\Gamma_2}+\abs{\Gamma_2^\perp}=q^s$ and $\abs{\Gamma_3}+\abs{\Gamma_3^\perp}=1$. This is reminiscent of the affine case, in which if we evaluate the traces corresponding to $a\in \mathcal{A}$ for the primary code, then for the dual code we do not need to consider the traces corresponding to $\II_{a'}$. The strategy in our case will be similar: for each $a\in \mathcal{A}$ such that we consider its traces in $B$, we will see that we do not consider the traces corresponding to $\II_{a'}$ in $D$. We start with the sets $\Gamma_1$ and $\Gamma_1^\perp$. 

\begin{lem}\label{suma1}
With the definitions as above, we have $\abs{\Gamma_1}+\abs{\Gamma_1^\perp}=q^2$.
\end{lem}
\begin{proof}
By definition, it is clear that we have $q^{2s}-\abs{\Gamma_1}=\sum_{a\in \mathcal{A}\setminus \mathcal{A}_{<d}}n_a$. We note that $a\in \mathcal{A}\setminus \mathcal{A}_{<d}$ if and only if there is $(c_1,c_2)\in \II_a$ such that $c_1+c_2\geq d$. Therefore, $2(q^s-1)-c_1-c_2=c'_1+c'_2\leq d^\perp$, which means that $M_{a'}(d^\perp)\neq \emptyset$. It is easy to see that $n_a=n_{a'}$, and we have $\sum_{a\in \mathcal{A}\setminus \mathcal{A}_{<d}}n_a=\sum_{a'\in \mathcal{A} \mid M_{a'}(d^\perp)\neq\emptyset}n_{a'}=\abs{\Gamma_1^\perp}$. Thus, $\abs{\Gamma_1}+\abs{\Gamma_1^\perp}=q^{2s}$.
\end{proof}

For the case of $\Gamma_2$ and $\Gamma_2^\perp$, we need the following technical results. 

\begin{lem}\label{lemaprimariomonomiosdostipos}
Let $1\leq d\leq 2(q^s-1)$, $d^\perp=2(q^s-1)-d$ and $c\in \mathcal{A}$. Then $M_{c'}(d^\perp)$ contains monomials of the two types if and only if $\II_c\cap (\Delta_d\cup \Delta_{2(q^s-1)-\overline{d^\perp}})\neq \emptyset$ and $\II_c\not\subset \Delta_{\leq d}$, where $\Delta_z=\emptyset$ if $z<0$.
\end{lem}
\begin{proof}
By Lemma \ref{condicionesM}, $M_{c'}(d^\perp)$ contains monomials of the two types if and only if $\II_{c'}\cap \Delta_{< d^\perp}\neq \emptyset$ and $\II_{c'}\cap (\Delta_{d^\perp}\cup \Delta_{\overline{d^\perp}})\neq \emptyset$. The condition $\II_{c'}\cap \Delta_{< d^\perp}\neq \emptyset$ implies that there is $(\gamma'_1,\gamma'_2)\in \II_{c'}$ such that $2(q^s-1)-\gamma_1-\gamma_2<d^\perp\iff \gamma_1+\gamma_2>d$. Thus, $\gamma\in \II_c\not \subset\Delta_{\leq d}$. The condition $\II_{c'}\cap (\Delta_{d^\perp}\cup \Delta_{\overline{d^\perp}})\neq \emptyset$ implies that there is an element $(\gamma'_1,\gamma'_2)\in \II_{c'}$ with either $2(q^s-1)-\gamma_1-\gamma_2=d^\perp$ or $2(q^s-1)-\gamma_1-\gamma_2=\overline{d^\perp}$.  Hence, $\gamma\in \Delta_d\cup\Delta_{2(q^s-1)-\overline{d^\perp}}$.
\end{proof}

\begin{rem}\label{remdt}
It is easy to check that $2(q^s-1)-\overline{d^\perp}=d$ if $d\geq q^s-1$, and $2(q^s-1)-\overline{d^\perp}=d+q^s-1$ if $d\leq q^s-2$.
\end{rem}

The following result, among other things, relates the set 
\begin{equation}\label{defY}
Y=\left\{a_2\in \mathcal{A}^1_{\leq d},\II_{a_2}\neq \II_{\overline{d}}\mid \bigcup_{c_2\in \II_{a_2},c_2>d-(q^s-1)}\II_{(d-c_2,c_2)}\subset \Delta_{\leq d}\right\}
\end{equation}
with the set $V=\{a_2\in \mathcal{A}^1\mid \II_{a_2}\neq \II_{\overline{d^\perp}} \text{ and } \exists \; c\in\mathcal{A}\text{ with } c_2=a_2 \text{ and } M_c(d^\perp) \text{ contains}$ $\text{ monomials of the two types}\}.$

\begin{lem}\label{lemana2}
Let $a_2\in \mathcal{A}^1_{\leq d}$. Then $\condicion$ if and only if there is no  $c\in \mathcal{A}$ with $\II_{c'}\neq \II_{(0,\overline{d^\perp})}$, $c'_2\in \II_{a'_2}$, and such that $M_{c'}(d^\perp)$ contains monomials of the two types.
\end{lem}
\begin{proof}
Let $a_2\in \mathcal{A}^1_{\leq d}$. By Lemma \ref{lemaprimariomonomiosdostipos}, we can translate the statement to the following: we have $\condicion$ if and only if there is no $c\in \mathcal{A}$, $\II_c\neq \II_{(q^s-1,\overline{d^\perp}')}$, with $c_2=a_2$, $\II_c\cap (\Delta_d\cup \Delta_{2(q^s-1)-\overline{d^\perp}})\neq \emptyset$ and $\II_c\not\subset \Delta_{\leq d}$. In what follows, we will use this last statement instead of the original one. We also note that $\overline{d^\perp}'=\overline{d}$ if $d\neq q^s-1$, and $\overline{d^\perp}'=0$ if $d=q^s-1$. 

We assume that $\condicion$ and we consider $c\in \mathcal{A}$, $\II_c\neq \II_{(q^s-1,\overline{d^\perp}')}$, with $c_2=a_2$. If $\II_c\cap \Delta_d\neq \emptyset$, we have $(d-\gamma_2,\gamma_2)\in \II_c$ for some $\gamma_2\in \II_{a_2}$. This implies that $d-\gamma_2\leq q^s-1$, i.e., $\gamma_2\geq d-(q^s-1)$. If $\gamma_2>d-(q^s-1)$, then, by our assumptions, $\II_c=\II_{(d-\gamma_2,\gamma_2)}\subset \Delta_{\leq d}$. If we had $\gamma_2=d-(q^s-1)$ and $d\geq q^s$, then this would imply that $(q^s-1,\overline{d})\in \II_c$, which is a contradiction with the fact that $\II_c\neq \II_{(q^s-1,\overline{d})}$. If $d=q^s-1$, then $\gamma_2=0$, which implies that $(d-\gamma_2,\gamma_2)=(q^s-1,0)$ and $\II_c=\{(q^s-1,0)\}$, a contradiction with the fact that $\II_c\neq \II_{(q^s-1,\overline{d^\perp}')}=\II_{(q^s-1,0)}$.

On the other hand, if $\II_c\cap \Delta_d=\emptyset$ and $\II_c\cap \Delta_{2(q^s-1)-\overline{d^\perp}}\neq \emptyset$, we have $\gamma\in \II_c$ with $\gamma_1+\gamma_2=2(q^s-1)-\overline{d^\perp}$, and $\gamma_2\in \II_{a_2}$. Considering Remark \ref{remdt}, if $d\geq q^s-1$, this implies $\gamma\in \Delta_d$, a contradiction with the assumption $\II_c\cap \Delta_d=\emptyset$. If $d\leq q^s-2$, then we note that $\gamma_2\leq d$ since $a_2\in \mathcal{A}_{\leq d}$, and $\gamma_1\leq q^s-1$, which implies $\gamma_1+\gamma_2\leq d+q^s-1=2(q^s-1)-\overline{d^\perp}$. We can only obtain the equality if $\gamma_1=q^s-1$ and $\gamma_2=d$, which is a contradiction with the assumption $\II_c\neq \II_{(q^s-1,\overline{d})}$.

For the other implication, we assume now that there is no $c\in \mathcal{A}$, $\II_c\neq \II_{(q^s-1,\overline{d^\perp}')}$, with $c_2=a_2$, $\II_c\cap(\Delta_d\cup\Delta_{2(q^s-1)-\overline{d^\perp}})\neq \emptyset$ and $\II_c\not\subset \Delta_{\leq d}$. For each $\gamma_2\in \II_{a_2}$, with $\gamma_2>d-(q^s-1)$, there is an element $c\in \mathcal{A}$ such that $\II_c=\II_{(d-\gamma_2,\gamma_2)}$. Because of the ordering chosen for the elements in $\zmsdos$, we must have $c_2=a_2$. We clearly have $(d-\gamma_2,\gamma_2)\in \II_c\cap \Delta_d\neq \emptyset$. By our assumption, we must have $\II_c=\II_{(d-\gamma_2,\gamma_2)}\subset \Delta_{\leq d}$.
\end{proof}

\begin{rem}\label{remgamma2}
Lemma \ref{lemana2} implies the following. Let $a_2\in \mathcal{A}^1_{\leq d}$ with $\II_{a_2}\neq \II_{\overline{d}}$. Then $a_2\in Y$ if and only if there is no $c\in \mathcal{A}$ with $c'_2\in \II_{a'_2}$, $\II_{c'}\neq \II_{(0,\overline{d^\perp})}$, and such that $M_{c'}(d^\perp)$ contains monomials of the two types. 

Recalling that $\overline{d}'=\overline{d^\perp}$ if $d\neq q^s-1$, and $\overline{d}'=0$ if $d=q^s-1$, we see that if $d\neq q^s-1$, $\II_{a_2}\neq \II_{\overline{d}}$ together with $c'_2\in \II_{a'_2}$ already implies $\II_{c'}\neq \II_{(0,\overline{d^\perp})}$. For $d=q^s-1$, in the case $a_2=0$, we see that the previous statement says:  $0\in Y$ if and only if there is no $c\in \mathcal{A}$ with $c'_2\in \II_{q^s-1}$, $\II_{c'}\neq \II_{(0,q^s-1)}$, and such that $M_{c'}(q^s-1)$ contains monomials of the two types. However, $M_{(0,q^s-1)}(q^s-1)=\{x_2^{q^s-1}\}$ does not have monomials of the two types. Therefore, in this case we can also omit the condition $\II_{c'}\neq \II_{(0,\overline{d^\perp})}$. 

Thus, we have the following statement. Let $a_2\in \mathcal{A}^1_{\leq d}$ with $\II_{a_2}\neq \II_{\overline{d}}$. Then $a_2\in Y$ if and only if there is no $c\in \mathcal{A}$ with $c'_2\in \II_{a'_2}$ and such that $M_{c'}(d^\perp)$ contains monomials of the two types. 
\end{rem}

\begin{lem}\label{lemacerod}
Let $1\leq d\leq 2(q^s-1)$, $d^\perp=2(q^s-1)-d$. If $d\neq q^s-1$, then $M_{(0,\overline{d^\perp})}(d^\perp)$ contains monomials of the two types if and only if $\II_{(q^s-1,\overline{d})}\not\subset \Delta_{\leq d}$. 
\end{lem}
\begin{proof}
If $d^\perp\geq q^s$, then $M_{(0,\overline{d^\perp})}(d^\perp)$ contains monomials of the two types because $x_0^{q^s-1}x_2^{\overline{d^\perp}}$, $x_2^d\in M_{(0,\overline{d^\perp})}(d^\perp)$. In this case, we have $d\leq q^s-2$, which ensures that $\II_{(q^s-1,\overline{d})}\not\subset \Delta_{\leq d}$.

If $d^\perp\leq q^s-1$, $M_{(0,d^\perp)}(d^\perp)$ contains monomials of the two types if and only if $d^\perp$ is not the minimal element of $\II_{d^\perp}$. We have $(d^\perp)'=q^s-1-d^\perp=q^s-1-(2(q^s-1)-d)=d-(q^s-1)$. The condition $d^\perp\leq q^s-1$ implies that $d\geq q^s-1$. Taking into account the assumption $d\neq q^s-1$, we can assume now that $d>q^s-1$. Thus, $(d^\perp)'=\overline{d}$, and we obtain that $M_{(0,d^\perp)}(d^\perp)$ contains monomials of the two types if and only if $d^\perp$ is not the minimal element of $\II_{d^\perp}$, which happens if and only if $(d^\perp)'=\overline{d}$ is not the maximal element of $\II_{\overline{d}}$, which happens if and only if $\II_{(q^s-1,\overline{d})}\not\subset \Delta_{\leq d}$.
\end{proof}

\begin{lem}\label{suma2}
We have that $\abs{\Gamma_2}+\abs{\Gamma_2^\perp}= q^s$.
\end{lem}
\begin{proof}
We start with the following decomposition:
$$
    q^s=\sum_{a_2\in\mathcal{A}^1}n_{a_2}=\sum_{a_2\in \mathcal{A}^1_{\leq d},a_2\in Y,\II_{a_2}\neq \II_{\overline{d}}}n_{a_2}+\sum_{a_2\in \mathcal{A}^1_{\leq d},a_2\not\in Y,\II_{a_2}\neq \II_{\overline{d}}}n_{a_2}+\sum_{a_2\in\mathcal{A}^1\setminus \mathcal{A}^1_{\leq d},\II_{a_2}\neq\II_{\overline{d}}}n_{a_2}+n_{\overline{d}}.
$$
We recall that $\sum_{a_2\in \mathcal{A}^1_{\leq d},a_2\in Y,\II_{a_2}\neq \II_{\overline{d}}}n_{a_2}=\abs{B_2}$. We also recall the definition $V=\{a_2\in \mathcal{A}^1\mid \II_{a_2}\neq \II_{\overline{d^\perp}} \text{ and } \exists \; c\in\mathcal{A}\mid c_2=a_2 \text{ and } M_c(d^\perp) \text{ contains monomials of the two types} \}$. Let $a_2\in \mathcal{A}^1_{\leq d}$. By Remark \ref{remgamma2}, if $d\neq q^s-1$, we have that $a_2\in Y$ if and only if the minimal element of $\II_{a'_2}$ is not in $V$. 
Taking into account that $n_{a_2}=n_{a'_2}$, we have that
$$
\sum_{a_2\in \mathcal{A}^1_{\leq d},a_2\not\in Y,\II_{a_2}\neq \II_{\overline{d}}}n_{a_2}=\sum_{b'_2\in V\mid \II_{b_2}=\II_{a_2},a_2\in\mathcal{A}^1_{\leq d}}n_{b'_2}.
$$

If $d\geq q^s-1$, we have $\mathcal{A}^1_{\leq d}=\mathcal{A}^1$, and the only thing left to do is to consider the cyclotomic set $\II_{\overline{d}}$. However, if $d\leq q^s-2$, we can consider $a_2\in \mathcal{A}^1\setminus \mathcal{A}^1_{\leq d}$. We have that $d\leq q^s-2\iff d^\perp\geq q^s$, and $a_2\in \mathcal{A}^1\setminus \mathcal{A}^1_{\leq d}$ implies that there is $\gamma_2\in\II_{a_2}$ with $\gamma_2>d\iff \gamma'_2< \overline{d^\perp}$ in this case. Hence, we can consider $c=(\overline{d^\perp}-\gamma'_2,\gamma'_2)$, and we have that $\{x_0^{q^s-1}x_1^{\overline{d^\perp}-\gamma'_2}x_2^{\gamma'_2} , x_1^{d^\perp-\gamma'_2}x_2^{\gamma'_2}\}\subset M_c(d^\perp)$, which means that $M_c(d^\perp)$ contains monomials of the two types, and $\II_{a'_2}\neq \II_{\overline{d^\perp}}$, i.e., if we consider $b_2\in \mathcal{A}^1$ such that $\II_{b_2}=\II_{a'_2}$, we have $b_2\in V$. 

Reciprocally, if we consider $a_2\in \mathcal{A}^1$ and we have $c'\in\mathcal{A}$ such that $c'_2\in \II_{a'_2}\neq \II_{\overline{d^\perp}}$ and $M_{c'}(d^\perp)$ contains monomials of the two types, there is $(\gamma'_1,\gamma'_2)\in \II_c$ with $\gamma'_1+\gamma'_2=\overline{d^\perp}=d^\perp-(q^s-1)$, which means that $\gamma_1+\gamma_2=d+(q^s-1)$, with $\gamma_2\in \II_{a_2}$. If $\gamma_1<q^s-1$, then $\gamma_2>d$ and $a_2\in\mathcal{A}\setminus \mathcal{A}_{\leq d}$. If $\gamma_1=q^s-1$, then $\gamma_2=d$, a contradiction since in this case $\II_{a'_2}\neq \II_{\overline{d^\perp}}$ implies $\II_{a_2}\neq \II_d$.
 
Thus, we have obtained that 
$$
\sum_{a_2\in \mathcal{A}^1_{\leq d},a_2\not\in Y,\II_{a_2}\neq \II_{\overline{d}}}n_{a_2}+\sum_{a_2\in\mathcal{A}^1\setminus \mathcal{A}^1_{\leq d},\II_{a_2}\neq\II_{\overline{d}}}n_{a_2}=\sum_{b'_2\in V}n_{b'_2}=\abs{D_2}.
$$

We now focus on the cyclotomic set $\II_{\overline{d}}$. We use Lemma \ref{lemacerod}, as we are still in the case $d\neq q^s-1$. If $d<q^s-1$, we always have $\abs{\Gamma_2}=\abs{B_2}$ by definition, and we also have $\abs{\Gamma_3}=\abs{D_2}+n_{d}$ because $\{x_0^{q^s-1}x_2^{\overline{d^\perp}},x_2^{d^\perp}\}\subset M_{(0,\overline{d^\perp})}(d^\perp)$, i.e., $M_{(0,\overline{d^\perp})}(d^\perp)$ contains monomials of the two types. If $d>q^s-1$, we have $\abs{\Gamma_2}=\abs{B_2}+n_{\overline{d}}$ if and only if $M_{(0,d^\perp)}(d^\perp)$ does not have monomials of the two types, by Lemma \ref{lemacerod}, and $\abs{\Gamma_2}=\abs{B_2}$ otherwise. Thus, we have that $\abs{\Gamma_2}=\abs{B_2}+n_{\overline{d}}$ if and only if $\abs{\Gamma_2^\perp}=\abs{D_2}$, and $\abs{\Gamma_2}=\abs{B_2}$ if and only if $\abs{\Gamma_2^\perp}=\abs{D_2}+n_{\overline{d}}$. Hence, for $d\neq q^s-1$ we have proved that
$$
\abs{\Gamma_2}+\abs{\Gamma_2^\perp} =q^s.
$$

On the other hand, if $d=q^s-1$, the condition $\II_{a_2}\neq \II_{\overline{d}}=\II_{q^s-1}$ implies $\II_{a'_2}\neq \II_{0}$ instead of $\II_{a'_2}\neq \II_{\overline{d^\perp}}=\II_{q^s-1}$. For any $a_2\in \mathcal{A}^1_{\leq d}=\mathcal{A}^1$, $a_2\not \in \{0,q^s-1\}$, the previous relations between elements in $Y$ and elements in $V$ hold by Remark \ref{remgamma2}. For $a_2=0$ and $a_2=q^s-1$ we have that $M_{(0,q^s-1)}(q^s-1)$ and $M_{(q^s-1,0)}(q^s-1)$ are the only sets $M_c(d^\perp)$ with $c_2=0'$ or $c_2=(q^s-1)'$, respectively, such that $x_0$ does not divide all the monomials in $M_c(q^s-1)$, and none of them contains monomials of the two types. Hence, for $d=q^s-1$, we obtain that $0\not \in V$, and also that $\abs{D_2}=\sum_{a'_2\in V}n_{a'_2}$ since $M_{(0,q^s-1)}(q^s-1)$ does not have monomials of the two types, and there is no other $c\in \mathcal{A}$ with $c_2=q^s-1$ such that $M_c(q^s-1)$ contains monomials of the two types. On the other hand, for $d=q^s-1$ is easy to see that $0\in Y$. Moreover, for $d=q^s-1$ we have that $\mathcal{A}^1\setminus \mathcal{A}^1_{\leq d}=\emptyset$, and we have $\II_{(q^s-1,q^s-1)}\not\subset \Delta_{q^s-1}$, which means that $\abs{\Gamma_2}=\abs{B_2}=\sum_{a_2\in Y}n_{a_2}$. Summarizing all of this, we have
$$
\abs{\Gamma_2}+\abs{D_2}+n_{q^s-1}=q^s,
$$
because for any $a_2\in \mathcal{A}^1$, $a_2\not \in \{0,q^s-1\}$, we have that either $a_2\in Y$ or $a'_2\in V$ as before, and we have that $0\in Y$, $q^s-1\not \in Y$ and $q^s-1\not \in V$. Obviously, in this case $n_{q^s-1}=1$, and for $d=q^s-1$, looking at the definition of $\Gamma_2^\perp$ from (\ref{gamma2perp}), we see that $\abs{\Gamma_2^\perp}=\abs{D_2}+1$ (the previous argument shows that, in this case $D_3=\emptyset$). Therefore, $\abs{\Gamma_2}+\abs{\Gamma_2^\perp}=q^s$.
\end{proof}

\begin{lem}\label{suma3}
We have that $\abs{\Gamma_3}+\abs{\Gamma_3^\perp}=1$.
\end{lem}
\begin{proof}
Let $a_2\in\mathcal{A}^1$ such that $\II_{a_2}=\II_{\overline{d^\perp}}$. By Remark \ref{condicionnested}, we have that $\Gamma_3^\perp\neq \emptyset$ if and only if there is an element $c\in\mathcal{A}$ such that $c_2=a_2$, $\II_c\neq \II_{(0,\overline{d^\perp})}$, and $M_c(d^\perp)$ contains monomials of the two types. By Lemma \ref{lemana2}, this happens if and only if $\notcondicion$. By the definition of $\Gamma_3$, this happens if and only if $\Gamma_3=\emptyset$. The cardinality of these sets is 1 if they are nonempty, which implies that $\abs{\Gamma_3}+\abs{\Gamma_3^\perp}=1$.
\end{proof}

Now we state the main result of this section.

\begin{thm}\label{baseprmprimario}
Let $1\leq d\leq 2(q^s-1)$. The image by the evaluation map of the set 
$$
B=B_1\cup B_2 \cup B_3,
$$
with $B_1,B_2,B_3$ as defined in Lemmas \ref{lemab1}, \ref{lemab1b2} and \ref{lemab1b2b3}, respectively, forms a basis for the code $\PRM_d^\sigma(2)$.
\end{thm}
\begin{proof}
By Lemma \ref{lemab1b2b3}, we know that the image by the evaluation map of the set $B$ is in $\PRM_d^\sigma(2)$, and it is linearly independent. By Lemmas \ref{suma1}, \ref{suma2} and \ref{suma3}, we have that $\abs{B}+\abs{D}=\abs{B}+\dim \PRM_d^{\sigma,\perp}(2)=q^2+q+1=n$. Thus, $B$ is a maximal linearly independent set, and we obtain the result.
\end{proof}

\begin{rem}
The sets $B'_2$ and $B'_3$ obtained using Remarks \ref{remfacil} and \ref{remfacil2}, respectively, also satisfy that $B_1\cup B'_2\cup B'_3$ is a basis for $\PRM_d^\sigma(2)$.
\end{rem}

We have that $\PRM_d^\sigma(2)$ is a subcode of $\PRM_d(2)$. Thus, we should be able to obtain $\PRM_d^\sigma(2)$ as the evaluation of some set of homogeneous polynomials of degree $d$. In fact, in all the discussions leading to Lemmas \ref{lemab1}, \ref{lemab1b2} and \ref{lemab1b2b3}, we showed how to construct homogeneous polynomials with the same evaluation as the ones considered in Theorem \ref{baseprmprimario}. Concrete expressions for these homogeneous polynomials can be given, but they get considerably more involved than the expressions obtained for the polynomials in $B$.

\begin{ex}
Continuing with Example \ref{ejemploprimario3}, Theorem \ref{baseprmprimario} states that the image by the evaluation map of the set $B=B_1\cup B'_2 \cup B_3$ that we have constructed in those examples gives a basis for the code $\PRM_{21}^\sigma(2)$. Indeed, it can be checked with Magma \cite{magma} that the dimension of $\PRM_{21}^\sigma(2)$ is precisely 133 (the cardinality of $B$), and that the evaluation of the polynomials in $B$ is in $\PRM_{21}^\sigma(2)$.
\end{ex}

\begin{cor}\label{dimprimario}
Let $1\leq d\leq 2(q^s-1)$. We have the following formula for the dimension of $\PRM_d^\sigma(2)$:
$$
\dim(\PRM_d^\sigma(2))=\abs{B_1}+\abs{B_2}+\abs{B_3}=\displaystyle \sum_{a\in \mathcal{A}_{<d}} n_a + \sum_{a_2\in Y} n_{a_2}+\epsilon,
$$
where, if we consider $b_2\in \mathcal{A}^1$ with $\II_{b_2}=\II_{\overline{d}}$, then $\epsilon=n_{\overline{d}}+1$ if $\II_{(q^s-1,\overline{d})}\subset \Delta_{\leq d}$; $\epsilon=1$ if $\II_{(q^s-1,\overline{d})}\not \subset \Delta_{\leq d}$ and $\condicionnn$; and $\epsilon=0$ otherwise. 
\end{cor}

We have seen in Lemma \ref{suma3} that we have the evaluation of a polynomial with $x_2^d$ in its support in $\PRM_d^{\sigma}(2)$ if and only if we do not have the evaluation of $(x_0-1)(x_1-1)$ in $\PRM_d^{\sigma,\perp}(2)$. If we have  the evaluation of $(x_0-1)(x_1-1)$ in $\PRM_d^{\sigma,\perp}(2)$, this implies that $\PRM_d^{\sigma}(2)$ is a degenerate code, with a common zero at the coordinate associated to $[0:0:1]$ for all its vectors. However, if we only have one common zero, the code that we obtain after puncturing are still different than the ones obtained in the affine case. Nevertheless, if we obtain that all the points in $[\{0\}\times \{1\}\times \fqs]$ are common zeroes of the vectors in $\PRM_d^\sigma(2)$, then, after puncturing, we obtain a subfield subcode of an affine Reed-Muller code.  

The only parameter left to estimate is the minimum distance. For a code $C$ we denote its minimum distance by $\wt(C)$. For the code $\PRM_d^\sigma(2)$ we have the bound given by the minimum distance of $\PRM_d(2)$ (see \cite{sorensen}):
\begin{equation}\label{cotadmin}
\wt(\PRM_d^\sigma(2))\geq (q^s-t)q^{s(1-r)},
\end{equation}
where $d-1=r(q^s-1)+t$, $0\leq t<q^s-1$. This is the usual way to bound the minimum distance of a subfield subcode, for instance see \cite{sanjoseSSCPRS} for the subfield subcodes of projective Reed-Solomon codes. For the subfield subcodes of projective Reed-Muller codes, this bound is sharp in most of the cases that we have checked with Magma \cite{magma} ($q^s \leq 9$). For example, in Table \ref{tabla91} from Section \ref{secexamples}, the bound is sharp except for $d=2$, which corresponds to a degenerate code, and for $d=10$ (the bound is 8 instead of 9).

For the dual code $\PRM_d^{\sigma,\perp}(2)$, there is no straightforward bound for the minimum distance, as we see next. Given $C\subset \fqs^n$, if $C^q=C$, where we understand this as the component wise power of the code, we say that $C$ is Galois invariant. By \cite[Thm. 4]{bierbrauercyclic}, we have that $\Tr(C)=C^\sigma$. Writing Theorem \ref{delsarte} as $C^\perp \cap \fq^n=\Tr(C)^\perp$, we note that $C^{\perp,\sigma}=C^\perp \cap \fq^n=(C^\sigma)^\perp= C^{\sigma,\perp}$. Therefore, when $C$ is Galois invariant, we have 
$$
\wt(C^{\sigma,\perp})=\wt(C^{\perp,\sigma})\geq \wt(C^\perp).
$$
This bound has been used frequently in the affine case \cite{galindolcd,galindostabilizer}, but in the projective case we do not have Galois invariant codes in general and we do not have the previous bound, nor the equality between $\PRM_d^{\sigma,\perp}(m)$ and $\PRM_d^{\perp,\sigma}(m)$. 

\section{Codes over the projective space}\label{secpm}

In this section we want to deal with the case of $m$ variables, for $m>2$. We have seen that, for $m=2$, obtaining bases for the subfield subcodes is quite technical. Hence, we do not aspire to give explicit results in this section for the bases of the subfield subcodes of projective Reed-Muller codes with $m>2$, but we can show that all the basic ideas can be generalized to treat this case. First we give a universal Gröbner basis for the vanishing ideal of $P^m$, which was a fundamental tool for the previous section when $m=2$. With respect to the terminology for Gröbner bases, we refer the reader to \cite{cox}. Particular cases of the following result were already presented in \cite{decodingRMP,sanjoseSSCPRS}.

\begin{thm}\label{vanishingideal}
The vanishing ideal of $P^m$ is generated by:
$$
\begin{aligned}
I(P^m)=&\langle  x_0^2-x_0,x_1^{q^s}-x_1,x_2^{q^s}-x_2,\dots,x_m^{q^s}-x_m,(x_0-1)(x_1^2-x_1),\\
&(x_0-1)(x_1-1)(x_2^2-x_2),\dots,(x_0-1)\cdots(x^2_{m-1}-x_{m-1}),(x_0-1)\cdots(x_m-1) \rangle.
\end{aligned}
$$
Moreover, these generators form a universal Gröbner basis of the ideal $I(P^m)$, and we have that
$$
\ini(I(P^m))=\langle x_0^2,x_1^{q^s},x_2^{q^s},\dots,x_m^{q^s},x_0x_1^2,x_0x_1x_2^2,\dots,x_0x_1\cdots x_{m-1}^2,x_0x_1\cdots x_m \rangle .
$$
\end{thm}

\begin{proof}
We consider the polynomials $f_0=x_0^2-x_0$, $f_1=x_1^{q^s}-x_1$, $f_2=x_2^{q^s}-x_2$,$\dots$, $f_m=x_m^{q^s}-x_m$, and $g_1=(x_0-1)(x_1^2-x_1)$, $g_2=(x_0-1)(x_1-1)(x_2^2-x_2)$,$\dots$, $g_{m-1}=(x_0-1)(x_1-1)\cdots (x_{m-2}-1) (x_{m-1}^2-x_{m-1})$, $g_m=(x_0-1)\cdots(x_m-1)$, and set $J:=\langle f_0,\dots,f_m,g_1,\dots, g_m\rangle$. 

Due to the generators $f_i$, $i=0,1,\dots,m$, it is clear that the variety defined by $J$ over the algebraic closure $\overline{\F_{q^s}}$ is the same as the variety defined over $\F_{q^s}$. By using \cite[Thm. 2.3]{ghorpadenull}, if we prove that the variety defined by $J$ over $\fqs$ is $P^m$, then we can conclude that $J=I(P^m)$. 

Given $P\in P^m$, we have that $P=[0:0:\dots:0:1:P_{l+1}:\dots:P_m]$ for some $l$, $0\leq l\leq m$, with $P_i\in\F_{q^s}$ for $i=l+1,\dots,m$. One can check that each generator of $J$ vanishes at $P$, which means that $P^m$ is contained in the variety defined by $J$.

Conversely, if all the generators of $J$ vanish at a point $P=[P_0:P_1:\dots:P_m]$, because of the generator $f_0$ the first coordinate is either 0 or 1. Considering the generator $g_m$, we also have that
$$
(P_0-1)(P_1-1)\cdots (P_m-1)=0.
$$
This means that there is an integer $l$ such that $P_l=1$, and we choose this $l$ to be the smallest with that property. If $l=0$, then $P=[1:P_1:\cdots :P_m]\in P^m$. If $l>0$, using the generator $g_{l-1}$ we obtain
$$
(P_0-1)(P_1-1)\cdots (P_{l-1}^2-P_{l-1})=0.
$$
Hence, $P_{l-1}=0$ since $P_0,P_1,\dots,P_{l-1}$ are different from 1 due to the choice of $l$. Doing this recursively we get that $P_0=P_1=\dots=P_{l-1}=0$, which means that $P=[0:0:\dots:0:1:P_{l+1}:\dots :P_m]\in P^m$. Therefore, we have $J=I(P^m)$. 

The only thing left to prove is that the generators of $I(P^m)$ form a universal Gröbner basis for $I(P^m)$. For any monomial order we have that $x_i>1$, $i=0,1,\dots,m$. Looking at each generator, we see that its initial monomial does not depend on the monomial order. Thus, if we prove that all the $S$-polynomials reduce to 0, and these reductions do not depend on the monomial order, we will have that these generators form a universal Gröbner basis for $I(P^m)$ using Buchberger's criterion \cite[\S 9 Thm. 3, Chapter 2]{cox}, and we will also obtain the stated initial ideal.

To show that all the $S$-polynomials reduce to 0, we will use two facts:
\begin{enumerate}
    \item[(a)] If the leading monomials of $f$ and $g$ are relatively prime, then $S(f,g)$ reduces to 0 by \cite[\S 9 Prop. 4, Chapter 2]{cox}. In particular, if $f$ and $g$ depend on different variables, then $S(f,g)$ reduces to 0. 
    \item[(b)] If $f$ and $g$ share a common factor $w$, then $S(f,g)=wS(f/w,g/w)$. Moreover, if we can apply (a) to $S(f/w,g/w)$, i.e., $S(f/w,g/w)$ reduces to 0 using $f/w$ and $g/w$, then $S(f,g)$ reduces to 0 using $f$ and $g$.  
\end{enumerate}

On one hand, for all $i,j$, $0\leq i<j\leq m$, we have that $S(f_i,f_j)$ reduces to $0$ by (a). On the other hand, for all $k,l$, $1\leq k<l< m$, using (b) we have
$$
\begin{aligned}
S(g_k,g_l)=(x_0-1)\cdots(x_{k-1}-1)(x_{k}-1)S(x_{k},(x_{k+1}-1)\cdots (x_{l-1}-1)(x_{l}^2-x_{l})),
\end{aligned}
$$
where the last $S$-polynomial reduces to 0 by (a). For $l=m$, the same argument applies, as we have
$$
\begin{aligned}
S(g_k,g_m)=(x_0-1)\cdots(x_{k-1}-1)(x_{k}-1)S(x_{k},(x_{k+1}-1)\cdots (x_{m-1}-1)(x_{m}-1)).
\end{aligned}
$$

Finally, we consider $S(f_i,g_k)$, for $1\leq i\leq m$, $1\leq k <m$. If $i>k$, this $S$-polynomial reduces to 0 by (a). If $i=k$, using (b) we have
$$
\begin{aligned}
S(f_k,g_k)=(x_{k}^2-x_{k})S((1+x_{k}+\cdots+x_{k}^{q^s-2}),(x_1-1)\cdots (x_{k-1}-1)),
\end{aligned}
$$
and the last $S$-polynomial reduces to 0 by (a). If $i<k$, applying (b) we obtain
$$
\begin{aligned}
&S(f_i,g_k)=(x_{i}-1)S(x_i(1+x_i+\cdots+x_i^{q^s-2}),(x_1-1)\cdots(x_{i-1}-1)(x_{i+1}-1)\cdots (x_{k}^2-x_k)),
\end{aligned}
$$
where the last $S$-polynomial reduces to 0 by (a). For the cases with $i=0$ or $k=m$, an analogous reasoning proves that the $S$-polynomials reduce to 0. 
\end{proof}

\begin{rem}
If $q^s>2$, from the proof of Theorem \ref{vanishingideal} we also obtain that the universal Gröbner basis obtained in Theorem \ref{vanishingideal} is in fact the reduced Gröbner basis with respect to any monomial order. Moreover, the same happens for any subset of the generators given in Theorem \ref{vanishingideal} and the ideal that they generate. 
\end{rem}

Now we give a convenient basis for $S/I(P^m)$, and also we show how to express any monomial in $S/I(P^m)$ in terms of this basis, i.e., we give the result of using the division algorithm for any monomial with respect to the universal Gröbner basis from Theorem \ref{vanishingideal}.

\begin{lem}\label{basePm}
The set given by the classes of the following monomials 
$$\{x_1^{a_1}\cdots x_m^{a_m}, x_0x_2^{a_2}\cdots x_m^{a_m}, \dots,x_0x_1\cdots x_{m-2}x_m^{a_m},x_0\cdots x_{m-1}\mid 0\leq a_i\leq q^s-1,1\leq i\leq m\}$$ 
is a basis for $S/I(P^m)$.
\end{lem}
\begin{proof}
Let $\mathcal{M}$ be the given set of monomials. We have that there is no monomial from $\mathcal{M}$ contained in $\ini(I(P^m))$ by Theorem \ref{vanishingideal}. We also have that $\abs{\mathcal{M}}=q^{sm}+q^{s(m-1)}+\cdots +q^s+1=\frac{q^{s(m+1)}-1}{q^s-1}=\abs{P^m}$, which is the dimension of $S/I(P^m)$ as a vector space (by definition, this is equal to $\deg(S/I(P^m))$, which is equal to $\abs{P^m}$ by \cite[Prop. 2.2]{jaramillo}). We finish the proof by noting that the classes of the monomials not contained in $\ini(I(P^m))$ form a basis for $S/I(P^m)$ \cite[Thm. 15.3]{eisenbud}.
\end{proof}

\begin{lem}\label{divisionPm}
Let $x_0^{a_0}x_1^{a_1}\cdots x_m^{a_m}=\prod_{i=0}^m x_i^{a_i}$ such that $a_0>0,a_1>0,\dots,a_l>0$ and $a_{l+1}=0$, with $0\leq l \leq m$ ($a_{k}:=0$ for $k>m$). Assume also that $a_i\leq q^s-1$, $1\leq i \leq m$. 
\begin{enumerate}[wide, labelwidth=!, labelindent=0pt]
\item[(a)] If $l< m$, then
\begin{dmath*}
\prod_{i=0}^m x_i^{a_i}\equiv \left(\prod_{i=l+2}^m x_i^{a_i} \right)\left( \prod_{i=1}^l x_i^{a_i}+(x_0-1)\left(\prod_{i=2}^l x_i^{a_i}+(x_1-1)\Bigg(\cdots \Bigg(x_l^{a_l}+(x_{l-1}-1)x_l  \Bigg)\cdots \Bigg) \right)  \right)\bmod I(P^m),
\end{dmath*}
where we understand that the product from $s$ to $t$ with $s>t$ is equal to 1. 
\item[(b)] If $l=m$, then
\begin{dmath*}
\prod_{i=0}^m x_i^{a_i}\equiv \left( \prod_{i=1}^m x_i^{a_i}\\
+(x_0-1)\left(\prod_{i=2}^m x_i^{a_i}+(x_1-1)\Bigg(\cdots \Bigg(x_m^{a_m}+(x_{m-1}-1) \Bigg)\cdots \Bigg) \right)  \right)\bmod I(P^m).
\end{dmath*}
\end{enumerate}
\end{lem}
\begin{proof}
Two polynomials belong to the same class in $S/I(P^m)$ if and only if their evaluation in $P^m$ is the same. Thus, to check the stated equivalences, it is enough to verify that both sides have the same evaluation in $P^m$. We assume first that $l<m$. We claim that
\begin{dmath*}
\prod_{i=0}^l x_i^{a_i}\equiv \prod_{i=1}^l x_i^{a_i}+(x_0-1)\left(\prod_{i=2}^l x_i^{a_i}+(x_1-1)\Bigg(\cdots \Bigg( x_l^{a_l}+(x_{l-1}-1)x_l\Bigg)\cdots \Bigg) \right) \bmod I(P^m).
\end{dmath*}
Indeed, if we decompose $P^m$ as in the proof of Lemma \ref{lemassc}, we can check that the evaluation of both sides is the same at each $A_r$, $0\leq r\leq m$. Because of the assumption $a_0>0$, the left hand side is $0$ at every point which is not in $A_0$. Both sides evaluate to the same values in $A_0$. For the evaluation in $A_r$, with $1\leq r <l$, we can set $x_0=x_1=\cdots=x_{r-1}=0$, and in the right hand side we get 
$$
(-1)^{r+1}\left( \prod_{i=r}^l x_i^{a_i}-\left( \prod_{i={r+1}}^l x_i^{a_i}+(x_r-1)\Bigg(\cdots \Bigg( x_l^{a_l}+(x_{l-1}-1)x_l\Bigg) \cdots \Bigg) \right)\right).
$$
Setting $x_r=1$, we obtain 0, which is what we get in the left hand side as well. If $r=l$, when we set $x_0=x_1=\cdots=x_{l-1}=0$ we obtain
$$
(-1)^{l+1}\left( x_l^{a_l}-x_l \right),
$$
which is equal to $0$ when we set $x_l=1$, as the left hand side. For $A_r$ with $l<r\leq m$, the right hand side is always 0 since it is divisible by $x_l$. Now (a) follows by considering the following factorization: 
$$
\prod_{i=0}^m x_i^{a_i}=\left(\prod_{i=l+2}^m x_i^{a_i} \right) \left( \prod_{i=0}^l x_i^{a_i}\right).
$$
An analogous argument shows that, when $l=m$, the polynomial stated in (b) has the same evaluation as $\prod_{i=0}^m x_i^{a_i}$ in $P^m$.
\end{proof}

\begin{rem}
It is not hard to see that all the monomials appearing in the right hand side of the expressions given in Lemma \ref{divisionPm} are part of the basis from Lemma \ref{basePm}.
\end{rem}

Hence, we have seen that the basic tools we have used for the case $m=2$ can be generalized to the case $m>2$. For the duals of the subfield subcodes, the reasoning that led to (\ref{trazadelatraza}) and (\ref{trazasgeneradoras}) shows that, in order to obtain a basis for $\T(S_d)$, for each monomial $x^\gamma\in S_d$, it is enough to consider the traces 
\begin{equation}\label{trazasgeneradorasm}
\{ \mathcal{T}_{\hat{\gamma}}(\xi_{\hat{\gamma}}^r x^\gamma)\mid 0\leq r\leq n_{\hat{\gamma}}-1\},
\end{equation}
where in this case we are considering cyclotomic sets in $m$ coordinates, and we extend the definitions for $\hat{\gamma}$ and $\mathcal{T}_{\hat{\gamma}}$ to this case in the obvious way. Hence, to obtain a basis we have to extract a maximal linearly independent set from the union of the previous sets. Theorem \ref{vanishingideal} and Lemma \ref{divisionPm} give the necessary tools to do that, but getting a general explicit formula for such a basis is quite involved.

For the primary code, the idea would be to consider homogenizations of the traces from the basis of the affine case from Theorem \ref{baseafin}, and then consider linear combinations of these polynomials such that, when setting $x_0=x_1=\cdots =x_j=0$ for some $0\leq j \leq m-1$, we obtain traces in less variables, similarly to what we did in the case of the projective plane.

\section{Examples}\label{secexamples}
In this section we show some examples of the parameters obtained from subfield subcodes of projective Reed-Muller codes over the projective plane. For computing the dimension, we can use Corollary \ref{dimdual} and Corollary \ref{dimprimario}, and for computing the minimum distance we use Magma \cite{magma}. We will denote the parameters of $\PRM_d^\sigma(2)$ by $[n,k,\delta]$, and the parameters of the dual code $\PRM_d^{\sigma,\perp}(2)$ by $[n,k^\perp,\delta^\perp]$. With respect to the parameters of the codes that we obtain, it is only possible to compare these codes with the codes from \cite{codetables} for small finite field sizes. This is because the codes that we obtain have length $n=\frac{q^{3s}-1}{q^s-1}=q^{2s}+q^s+1$, which gives rise to very long codes when we increase $q$ or $s$. Moreover, it is better to consider moderate values of $s$ due to the fact that the size of the corresponding cyclotomic sets increases for larger $s$, and therefore if we start with degree $d$ and we consider degree $d-1$, for each monomial of degree $d$ that we are no longer evaluating, all its powers of $q$ (seen in $S/I(P^2)$) will not appear in any trace from the basis that we have given for $\PRM_d^\sigma(2)$, and the size of the set formed by the monomial and its powers of $q$ is precisely the size of the corresponding cyclotomic set. This can cause significant drops in dimension, leading in some cases to codes with worse parameters compared to the cases with smaller $s$. Thus, we first consider binary codes and ternary codes arising from extensions of small degree.

For the extensions $\F_4\supset \F_2$ and $\F_8\supset \F_2$, we obtain the parameters from Table \ref{tablabinarios}. For the extension $\F_8\supset \F_2$ we omit the codes with $d=2,3$ as they are equal to $\PRM_1^{\sigma}(2)$. In the cases where $\delta^\perp$ is 1, we have that $\PRM_d(2)$ is a degenerate code. For instance, for the extension $\F_4\supset \F_2$, for $d=1$ we have $q^s+1=5$ common zeroes for all the vectors in the code, which means that, after puncturing, we obtain the same as the subfield subcode of an affine Reed-Muller code. However, for $d=2$ we only have $1$ common zero, and the corresponding code after puncturing does not correspond to the subfield subcode of any affine Reed-Muller code. With respect to the parameters, some of the codes from Table \ref{tablabinarios} have the best known parameters for a linear code with its length and dimension, according to \cite{codetables}. For example, that is the case for the codes with parameters  $[21,9,8]_2$, $[21,12,5]_2$ and $[21,16,3]_2$.

\begin{table}[ht]
\caption{Binary codes corresponding to the extensions $\F_4\supset \F_2$ and $\F_8\supset \F_2$, respectively.} % title of Table
\centering
%\begin{center}
\begin{tabular}{||c|c||c|c||c|c||}
 \hline % \hline
$d$ & $n$ & $k$ & $\delta$  & $k^\perp$ & $\delta^\perp$ \\
  \hline \hline
1 & 21 & 1 & 16 & 20 & 1 \\
2 & 21 & 2 & 12 & 19 & 1 \\
3 & 21 & 9 & 8 & 12 & 5 \\
4 & 21 & 11 & 4 & 10 & 2 \\
5 & 21 & 16 & 3 & 5 & 8 \\
6 & 21 & 20 & 2 & 1 & 21 \\

\hline
 %\hline
\end{tabular}
%\captionof{table}{Stabilizer affine variety ones codes over $\mathbb{F}_2$}
%\caption{Linear codes over $\mathbb{F}_4$ which are records}
%\end{center}
\begin{tabular}{||c|c||c|c||c|c||}
 \hline % \hline
$d$ & $n$ & $k$ & $\delta$  & $k^\perp$ & $\delta^\perp$ \\
  \hline \hline
1 & 73 & 1 & 64 & 72 & 1 \\
4 & 73 & 2 & 40 & 71 & 1 \\
5 & 73 & 7 & 32 & 66 & 1 \\
6 & 73 & 8 & 24 & 65 & 1 \\
7 & 73 & 27 & 16 & 46 & 9 \\
8 & 73 & 28 & 8 & 45 & 1 \\
9 & 73 & 32 & 8 & 41 & 2 \\
10 & 73 & 40 & 8 & 33 & 1 \\
11 & 73 & 51 & 5 & 22 & 16 \\
12 & 73 & 59 & 4 & 14 & 4 \\
13 & 73 & 66 & 3 & 7 & 32 \\
14 & 73 & 72 & 2 & 1 & 73 \\

\hline
 %\hline
\end{tabular}
\label{tablabinarios}
\end{table}

With respect to ternary codes, we consider the extension $\F_9\supset \F_3$. The parameters of the corresponding codes are presented in Table \ref{tabla91}, where we have omitted the case $d=2$ since it corresponds to the same code as $\PRM_1^\sigma(2)$. 
 
\begin{table}[ht]
\caption{Ternary codes corresponding to the extension $\F_9\supset \F_3$.} % title of Table
\centering
%\begin{center}
\begin{tabular}{||c|c||c|c||c|c||}
 \hline % \hline
$d$ & $n$ & $k$ & $\delta$  & $k^\perp$ & $\delta^\perp$ \\
  \hline \hline
1 & 91 & 1 & 81 & 90 & 1 \\
3 & 91 & 2 & 63 & 89 & 1 \\
4 & 91 & 9 & 54 & 82 & 4 \\
5 & 91 & 9 & 45 & 82 & 1 \\
6 & 91 & 10 & 36 & 81 & 1 \\
7 & 91 & 19 & 27 & 72 & 1 \\
8 & 91 & 36 & 18 & 55 & 10 \\
9 & 91 & 38 & 9 & 53 & 2 \\
10 & 91 & 45 & 9 & 46 & 4 \\
11 & 91 & 58 & 7 & 33 & 18 \\
12 & 91 & 70 & 6 & 21 & 36 \\
13 & 91 & 73 & 5 & 18 & 6 \\
14 & 91 & 80 & 4 & 11 & 36 \\
15 & 91 & 86 & 3 & 5 & 54 \\
16 & 91 & 90 & 2 & 1 & 91 \\

\hline
 %\hline
\end{tabular}
%\captionof{table}{Stabilizer affine variety ones codes over $\mathbb{F}_2$}
%\caption{Linear codes over $\mathbb{F}_4$ which are records}
\label{tabla91}
%\end{center}
\end{table}

We can compare the parameters of these codes with the ones obtained with affine Reed-Muller codes. Besides the fact that we obtain longer codes for the same field size, if we consider $\frac{k+\delta}{n}$ as a measure of how good a code is, we usually have that the projective code $\PRM_d^\sigma(2)$ is better in that sense than $\RM_d^\sigma(2)$. For example, we have that the code $\RM_4^\sigma(2)$ corresponding to the extension $\F_9\supset \F_3$ has parameters $[81,9,45]_3$, and $\PRM_4^\sigma(2)$ has parameters $[91,9,54]_3$, and one can check that $\PRM_4^\sigma(2)$ has better parameters with respect to the value $\frac{k+\delta}{n}$. In fact, the parameters of the code $\PRM_4^\sigma(2)$ are the best known parameters for a code with length $91$ and dimension $9$ over $\F_3$, according to \cite{codetables}. Moreover, the codes from Table \ref{tabla91} with parameters $[91,21,36]_3$, $[91,82,4]_3$ and $[91,86,3]_3$ are also the best known according to \cite{codetables}. 

For extensions of higher degree, or for fields with higher $q$, the codes that we obtain in this way are too long to be compared to the ones from \cite{codetables}. As we have seen in the previous examples, some of the codes that we obtain have the best known parameters, while others do not have great parameters. Focusing on the ones with better parameters, in Table \ref{tablaGV} we provide some long codes that surpass the Gilbert-Varshamov bound for different field extensions. For the minimum distance, we use the bound (\ref{cotadmin}) since these codes are too large for Magma \cite{magma}.

\begin{table}[ht]
\caption{Long codes exceeding the Gilbert-Varshamov bound.} % title of Table
\centering
%\begin{center}
\begin{tabular}{||c|c|c||c|c|c||}
 \hline % \hline
 $q$ & $s$ &$d$ & $n$ & $k$ & $\delta \geq $  \\
  \hline \hline
 2 & 4 & 28 & 273 & 255& 4\\
2 & 4 & 29 & 273 & 264&3\\
 4 & 2& 5& 273 & 9 & 192 \\
 4 & 2& 28& 273 & 262 & 4 \\
 4 & 2& 29& 273 & 268 & 3 \\
5 & 2 & 6 & 651 & 9 & 500\\
5 & 2 & 46 & 651 & 640 & 4\\
5 & 2 & 47 & 651 & 646 & 3\\
3 & 3 & 50 & 757 & 741 & 4 \\
3 & 3 & 51 & 757 & 750 & 3 \\
2 & 5 & 60 & 1057 & 1035& 4\\
2 & 5 & 61 & 1057 & 1046& 3\\
7 & 2 & 8 & 2451 & 9 & 2058 \\
7 & 2 & 94 & 2451 & 2440 & 4 \\
7 & 2 & 95 & 2451 & 2446 & 3 \\
\hline
 %\hline
\end{tabular}
%\captionof{table}{Stabilizer affine variety ones codes over $\mathbb{F}_2$}
%\caption{Linear codes over $\mathbb{F}_4$ which are records}
\label{tablaGV}
%\end{center}
\end{table}

Finally, for the case $m>2$, in Table \ref{tablam3} we show the binary codes obtained by considering the subfield subcodes of projective Reed-Muller codes over $\PP^3$ with respect to the extension $\F_4\supset \F_2$, where we have computed the parameters with Magma \cite{magma}. The codes with parameters $[85,16,32]_2$, $[85,60,8]_2$ and $[85,78,3]_2$ have the best known parameters according to \cite{codetables}.

\begin{table}[ht]
\caption{Binary codes corresponding to the extension $\F_4\supset \F_2$ with $m=3$.} % title of Table
\centering
%\begin{center}
\begin{tabular}{||c|c||c|c||c|c||}
 \hline % \hline
$d$ & $n$ & $k$ & $\delta$  & $k^\perp$ & $\delta^\perp$ \\
  \hline \hline
1 & 85 & 1 & 64 & 84 & 1 \\
2 & 85 & 2 & 48 & 83 & 1 \\
3 & 85 & 16 & 32 & 69 & 5 \\
4 & 85 & 18 & 16 & 67 & 1 \\
5 & 85 & 33 & 12 & 52 & 2 \\
6 & 85 & 60 & 8 & 25 & 21 \\
7 & 85 & 67 & 4 & 18 & 8 \\
8 & 85 & 78 & 3 & 7 & 32 \\
9 & 85 & 84 & 2 & 1 & 85 \\
\hline
 %\hline
\end{tabular}
%\captionof{table}{Stabilizer affine variety ones codes over $\mathbb{F}_2$}
%\caption{Linear codes over $\mathbb{F}_4$ which are records}
\label{tablam3}
%\end{center}
\end{table}

\bibliographystyle{abbrv}
%\bibliography{BIBR}

\end{document}